\documentclass[english]{article}
\usepackage{geometry}
\usepackage{float}
\usepackage{mathrsfs}
\usepackage{amsmath}
\usepackage{amssymb}
\usepackage{graphicx}
\usepackage{amsfonts}
\usepackage{bm}
\usepackage{subfigure}
\usepackage{amsthm}
\usepackage{epsfig}
\makeatletter

\@ifundefined{definecolor}{\@ifundefined{definecolor}
 {\@ifundefined{definecolor}
 {\usepackage{color}}{}
}{}
}{}

\usepackage{subfig}\usepackage[all]{xy}

\newtheorem{rem}{Remark}[section]
\newtheorem{prop}{Proposition}[section]

\newcounter{hypA}
\newenvironment{hypA}{\refstepcounter{hypA}\begin{itemize}
  \item[({\bf A\arabic{hypA}})]}{\end{itemize}}
\newcounter{hypB}

\usepackage{babel}\date{}

\usepackage{babel}

\makeatother

\usepackage{babel}

\begin{document}

\begin{center}

{\Large \textbf{Computational Methods for a Class of Network Models}}

\vspace{0.5cm}

BY JUNSHAN WANG$^{1}$, AJAY JASRA$^{1}$ \& MARIA DE IORIO$^{2}$

{\footnotesize $^{1}$Department of Statistics \& Applied Probability,
National University of Singapore, Singapore, 117546, SG.}\\
{\footnotesize E-Mail:\,}\texttt{\emph{\footnotesize a0082738@nus.edu.sg staja@nus.edu.sg}}\\
{\footnotesize $^{2}$Department of Statistical Science,
University College, London, WC1E 6BT, UK.}\\
{\footnotesize E-Mail:\,}\texttt{\emph{\footnotesize m.deiorio@ucl.ac.uk}}
\end{center}

\begin{abstract}

In the following article we provide an exposition of exact computational methods to perform parameter
inference from partially observed network models. In particular, we consider the duplication attachment (DA)
model which has a likelihood function that typically cannot be evaluated in any reasonable computational time.
We consider a number of importance sampling (IS) and sequential Monte Carlo
(SMC) methods for approximating the likelihood of the network model for a fixed parameter value.
It is well-known that for IS, the relative variance of the likelihood estimate typically grows at an exponential rate in the time
parameter (here this is associated to the size of the network): we prove that, under assumptions, the SMC method will have relative variance which can grow only polynomially. In order to perform parameter estimation, we develop
particle Markov chain Monte Carlo (PMCMC) algorithms to perform Bayesian inference. Such algorithms use the afore-mentioned
SMC algorithms within the transition dynamics. The approaches are illustrated numerically.
\end{abstract}

\section{Introduction}

Cellular functions are based on the complex interplay of proteins, therefore understanding the structure and dynamics of these protein-protein interaction (PPI) networks is paramount to gain insight into biological systems. Proteins are at the heart of the relationship between genotype and phenotype and the last years have witnessed large investments to investigate large-scale PPI networks of several model organisms. As a consequence, a significant amount of data has been collected and extensive studies on protein interaction networks have been carried out due not only to technical advances but also to developments in bioinformatic and statistical methods. Probabilistic models are indispensable for characterising the process of protein evolution  and are particularly valuables as they provide a sound basis for likelihood-based inference, as an alternative to statistical analysis based on summary statistics. A number of theoretical models have been developed to explain both the network formation, evolution  and  current structure. A popular class of mathematical models includes the duplication-attachment (DA) models, which specify probability distribution for the inclusion of new nodes and edges in the network. Therefore the network becomes the result of an evolutionary stochastic process, where the number of nodes in the network has increased though a series of node adding events. The probability distributions that govern network evolution depend on a vector of unknown parameter which are usually the main focus of statistical inference. More technical details will be given in section \ref{sec:model}.
  
This article contains an exposition of the challenges associated to parameter inference for network models and it focuses on exact computational methods. The approaches proposed are a valuable alternative to heuristic but fast analysis based on summary statistics, such as approximate Bayesian computation, which often provide approximations of the posterior distributions of the parameter 
which are not well understood. The class of models we consider, have a likelihood,
associated to an observed network  $G_t$ of $t$ vertices, with parameters $\theta\in\Theta\subseteq\mathbb{R}^d$, $d\geq 1$.  
The likelihood can be written as the expectation w.r.t~a probability that sequentially removes the vertices of the network until some terminal state is reached:
$$
L_{\theta}(G_t) = \mathbb{E}_{\theta}\Big[\prod_{k=0}^{t-t_0-1}\Big( H_k(\nu_{k:t};G_t)\Big) \Big]
$$
where $\nu_k$ is a vertex of  $G_t$,  $\mathsf{V}\subseteq\mathbb{N}$ is the vertex set, $\mathsf{E}$ is the edge set,  $H_k:\mathsf{V}^{t-k+1}\times\mathsf{E}\times\mathsf{V}^t\rightarrow\mathbb{R}_+$,
$t-t_0$ is the deterministic number of steps associated to the terminal state 
and the probability associated to the expectation can be written $\prod_{k=0}^{t-t_0-1} p_{\theta}(\nu_k|\nu_{k+1:t})$, $\nu_{k+1:t} = (\nu_{k+1},\dots,\nu_t)$ (if $t\geq k+1$, otherwise it is the null vector). We make the model and terminology precise in Section \ref{sec:model}; see Wiuf et al.~(2006) for examples of such models.
In most scenarios of practical interest one cannot compute the likelihood exactly, unless $t-t_0$ is very small (e.g.~$10$); direct calculation, for a single $\theta$, is an $\mathcal{O}((t-t_0) 2^{t-t_0})$ operation at best - see Wiuf et al.~(2006). Often, in the literature, one resorts to numerical methods based upon Monte Carlo and particularly importance sampling. This
procedure introduces a mutually absolutely continuous probability with joint mass function $\prod_{k=0}^{t-t_0-1} q_{\theta}(\nu_k|\nu_{k+1:t})$ (and associated expectation operator
$\mathbb{E}_{\theta}^q$) and then uses the simple change of measure formula
$$
L_{\theta}(G_t) = \mathbb{E}_{\theta}^q\Big[\prod_{k=0}^{t-t_0-1} \Big(H_k(\nu_{k:t};G_t)w_{k,\theta}(\nu_{k:t})\Big) \Big]
$$
where $w_{k,\theta}(\nu_{k:t}) = p_{\theta}(\nu_k|\nu_{k+1:t})/q_{\theta}(\nu_k|\nu_{k+1:t})$. This idea is adopted in Wiuf et al.~(2006) (see also Guetz \& Holmes (2011), which we discuss later in the article) and indeed, Wiuf et al.~(2006) choose a particularly clever proposal, based upon the optimal importance distribution e.g.~Robert \& Casella (2004). However, as is well-known in the literature, such IS procedures,
with estimates of the likelihood:
$$
\frac{1}{N} \sum_{i=1}^N \prod_{k=0}^{t-t_0-1} \Big(H_k(\nu_{k:t}^i;G_t)w_{k,\theta}(\nu_{k:t}^i)\Big)
$$ 
with $\{\nu_{k:t}^i\}_{1\leq i \leq N}$ sampled i.i.d.~from $q_{\theta}$, are known to perform extremely badly in practice; often the relative variance of such estimates are $\mathcal{O}(\kappa^{t-t_0})$, with $\kappa>1$ (see e.g.~C\'erou et al.~(2011) or Whiteley et al.~(2012)). One method which deals with this issue, at least for some classes of models, is that of sequential Monte Carlo methods.
This algorithm generates a collection of samples (also called particles) in parallel, using the same ideas as IS, except when the weights are too variable, in some sense, the samples are sampled with replacement
from the current particle set and weights reset to 1; see Doucet \& Johansen (2011) for an introduction, and we describe this algorithm in details in Section \ref{sec:comp}. For some classes of models,
estimates of quantities such as $L_{\theta}(G_t)$ have a relative variance of $\mathcal{O}(t-t_0)$; these results are extended for the network models considered in this article and we show that the relative variance will grow only polynomially in $t-t_0$ (Proposition \ref{prop:rel_var}).
 In addition, we consider a more advanced SMC method called the discrete
particle filter (DPF) (Fearnhead, 1998) and illustrate its applicability for likelihood estimation for the given class of network models.

The discussion so far has focussed upon estimation of the likelihood for a single $\theta\in\Theta$. To infer the parameter, we will follow a Bayesian procedure and place a prior probability
distribution on the parameter; we will then seek to sample from the associated posterior distribution using MCMC. This is a particularly challenging problem, as the `obvious'
idea of sampling a posterior proportional to
$$
\prod_{k=0}^{t-t_0} \Big(H_k(\nu_{k:t};G_t)\Big) \prod_{k=0}^{t-t_0} p_{\theta}(\nu_k|\nu_{k+1:t}) p(\theta) 
$$
where $p(\theta)$ is the prior on $\theta$, is very complex (see the discussion in e.g.~Andrieu et al.~(2010) for simpler models). However, one algorithm has been developed for these class of problems in Andrieu et al.~(2010). This method can use any of the SMC
or DPF methods within the proposal mechanism, and we develop such algorithms in Section \ref{sec:comp}.

This article is structured as follows. In Section \ref{sec:model} we consider the model and likelihood for a test model which is considered throughout the article. In Section \ref{sec:comp}
we review and develop computational methods for these network models. This is split into two types; one for approximating the likelihood for a fixed parameter (using SMC) and the other, which uses these afore-mentioned methods, to infer parameters in Bayesian way and is based upon MCMC. 
In Section \ref{sec:comp} we also give our relative variance result for the SMC method.
In Section \ref{sec:numerics} the approaches are numerically implemented and compared on small to medium sized networks. In Section \ref{sec:summary} the article is concluded and some avenues for future work are discussed.
The proof of our main result can be found in the appendix.

\section{Modelling and Likelihood}\label{sec:model}

\subsection{Model}

In the following discussion, we will be dealing with random graphs, so some probabilistic notations are introduced below.
Let $G_t$ be an undirected graph of $t$ vertices, without multiple
edges or self loops. 
DA models are essentially growth models, where the graph only expands through time and never loses nodes or edges. At each time step, a node is added through either a duplication or an attachment event. In the case of an attachment, the new node attaches to an old one, chosen uniformly over the graph, while in the case of a duplication event, an existing node is randomly picked to duplicate, and the new node is created by copying with a certain probability  each of its model's links (independently). The new node is then linked to the existing graph with a specified probability and it is, therefore,  possible that it is added with no links to any existing node. Let $\mathsf{V}\subseteq\mathbb{N}$ be the vertex set with associated $\sigma-$algebra $\mathscr{V}$ let $\mathsf{E}$ be the edge set with associated
$\sigma-$algebra $\mathscr{E}$. Given such a graph $G_t$, let $\delta(G_t,\nu)$,  $\nu\in\mathsf{V}$ denote the graph with $\nu$ deleted (i.e.~both the node and its associated edges). A vertex $\nu$ is said to be \emph{removable} if $G_t$ can be created by copying a vertex in $\delta(G_t,\nu)$. If $G_t$ contains removable nodes it is said to be
\emph{reducible}, otherwise it is \emph{irreducible}.

Let $\theta\in\Theta$ and consider a stochastic process $\{G_t\}_{t\in\{t_0,t_0+1,\dots\}}$ on probability space $(\Omega,\mathscr{F},\mathbb{P}_{\theta})$, with $\Omega = (\mathsf{V}\times\mathsf{E})^{\mathbb{N}}$, $\mathscr{F}=(\mathscr{V}\otimes\mathscr{E})^{\mathbb{N}}$, where for each
$A\in\mathscr{F}$, $\mathbb{P}_{\theta}(A)$ is $\mathcal{B}(\Theta)-$measurable ($\mathcal{B}(\Theta)$ are the borel sets on $\Theta$). A duplication-attachment model
is such a stochastic process that starts at an irreducible graph $G_{t_0}$ and undergoes Markov transitions according to a
transition probability which is only non-zero if $G_{t+1}$ can be obtained by copying a vertex in $G_t$.
In Wiuf et al.~(2006) it is stated that the likelihood associated to a given (reducible) graph $G_t$ can be written in the recursive
manner:
$$
L_{\theta}(G_t) = \frac{1}{t}\sum_{\nu\in\mathcal{R}(G_t)} \omega_{\theta}(\nu,G_t) L_{\theta}(\delta(G_t,\nu))
$$
with $L_{\theta}(G_{t_0})=1$, $\omega_{\theta}(G_t,\nu)=\mathbb{P}_{\theta}(G_t|\delta(G_t,\nu))$ the transition probability and $\mathcal{R}(G_t)$ is the collection of removable
vertices of $G_t$.

Let $\theta=(\pi, p, q, r)$, the DA model follows two transition rules.
\begin{enumerate}
\item Choose and duplicate $\nu_{old}$ in $G_t$ uniformly and call the copy as $\nu_{new} $. Create a link between $\nu_{new}$ and any node that links to $\nu_{old}$ with probability $p$. Link $\nu_{old}$ to $\nu_{new} $ with probability $q$.
 \item Choose and duplicate $\nu_{old}$ in $G_t$ uniformly and call the copy as $\nu_{new} $. Create a link between $\nu_{old}$ and $\nu_{new}$ with probability $r$.
\end{enumerate}
In every transition, we follow rule 1 with probability $\pi$, and rule 2 with probability $1-\pi$.

\subsection{Likelihood Computation}

One observes a reducible graph $G_t$, for which it is possible to obtain an irreducible graph $G_{t_0}$. Define, $\mathsf{V}(G_t)$ as the vertex set of $G_t$ and
for $k\in\{1,1,\dots,t-t_0-1\}$:
$$
\mathsf{V}_{t-k:t}:=\{\nu_{t-k:t}\in\mathsf{V}(G_t):\nu_{t-k}\neq\nu_{t-k+1}\neq\cdots\neq \nu_t\}
$$
for $k\geq 1$
$$
\delta^k(G_t,\nu_{t-k+1:t}) = \delta(\delta^{k-1}(G_t,\nu_{t-k+1:t}),\nu_{t-k+1}) \quad \nu_{t-k+1:t} \in \mathsf{V}_{t-k:t}
$$
with $\delta^0(G_t,\nu_{t+1}):=G_t$ and $\mathsf{V}_{k}:=\mathsf{W}_k(\nu_{t-k+1:t})=\mathcal{R}(\delta^k(G_t,\nu_{t-k+1:t}))$ ($k\in\{1,\dots,t-t_0\}$).
Then it follows that:
\begin{equation}
L_{\theta}(G_t) = \sum_{\nu_{t_0+1:t}\in\mathsf{V}_{t_0+1:t}}\bigg[
\prod_{k=0}^{t-t_0+1}\bigg\{
\frac{\mathbb{I}_{\mathsf{W}_{k}(\nu_{t-k+1:t})}(\nu_{t-k})\omega_{\theta}(\delta^k(G_t,\nu_{t-k+1:t}),\nu_{t-k})}{t-k}
\bigg\}
\bigg].\label{eq:like_vertex}
\end{equation}

\section{Computational Methods}\label{sec:comp}

We now consider a collection of numerical (SMC) techniques, first to approximate the likelihood \eqref{eq:like_vertex} for one $\theta$ and then a method to sample from the marginal posterior
$$
\pi(\theta|G_t) \propto L_{\theta}(G_t) p(\theta)
$$
which uses these SMC approaches.

\subsection{Importance Sampling}\label{sec:is}

The underlying idea of Wiuf et al.~(2006) is to introduce a probability $q_{\theta_0}(\nu_{t_0+1:t})$ on $\mathsf{V}_{t_0+1:t}$, with $\theta_0\in\Theta$, such
that $q_{\theta_0}(\nu_{t_0+1:t})>0$ for each $\nu_{t_0+1:t}\in\mathsf{V}_{t_0+1:t}$. 
Note that $\theta_0$ is termed a driving value and can help one approximate $L_{\theta}(G_t)$ for many $\theta$; see Griffiths \& Tavar\'e (1994). 
The proposal $q_{\theta_0}$ is decomposed as
$$
q_{\theta_0}(\nu_{t_0+1:t})  = \prod_{k=0}^{t-t_0-1} q_{\theta_0}(\nu_{t-k}|\nu_{t-k+1:t})
$$
and hence we have the change of measure formula
$$
L_{\theta}(G_t) = \mathbb{E}^q_{\theta_0}\bigg[
\prod_{k=0}^{t-t_0-1}\bigg\{
\frac{\mathbb{I}_{\mathsf{W}_{k}(\nu_{t-k+1:t})}(\nu_{t-k})\omega_{\theta}(\delta^k(G_t,\nu_{t-k+1:t}),\nu_{t-k})}{(t-k)q_{\theta_0}(\nu_{t-k}|\nu_{t-k+1:t})}
\bigg\}
\bigg]
$$
where $\mathbb{E}^q_{\theta_0}$ is the expectation w.r.t.~$q_{\theta_0}(\nu_{t_0+1:t})$.
The estimate of the likelihood is then:
\begin{equation}
\frac{1}{N} \sum_{i=1}^N \prod_{k=0}^{t-t_0-1}\bigg\{
\frac{\mathbb{I}_{\mathsf{W}_{k}(\nu_{t-k+1:t}^i)}(\nu^i_{t-k})\omega_{\theta}(\delta^k(G_t,\nu^i_{t-k+1:t}),\nu^i_{t-k})}{(t-k)q_{\theta_0}(\nu_{t-k}^i|\nu^i_{t-k+1:t})}
\bigg\}
\label{eq:like_is_exp}
\end{equation}
with $\{\nu^i_{t_0+1:t}\}_{1\leq i \leq N}$ sampled i.i.d.~from the probability $q_{\theta_0}$.

The idea is that given $G_t$ and the ordered list
$\nu_{t-k:t}$ one can easily determine the graph that is obtained after removing $k+1$ vertices; so one can sample the vertices.
A clear point, not mentioned by Wiuf et al.~(2006), is that the conditionally optimal importance sampling proposal 
(that is minimizing the variance of the importance weights, given $\nu_{t-k+1:t}$)
is
\begin{equation}
q_{\theta}(\nu_{t-k}|\nu_{t-k+1:t}) \propto\mathbb{I}_{\mathsf{W}_{k}(\nu_{t-k+1:t})}(\nu_{t-k})\omega_{\theta}(\delta^k(G_t,\nu_{t-k+1:t}),\nu_{t-k})
\label{eq:opt_proposal}
\end{equation}
although, the authors use this proposal in their simulations.
As noted in the introduction, estimates such as \eqref{eq:like_is_exp} often have a relative
variance that is $\mathcal{O}(\kappa^{t-t_0})$, $\kappa>1$. 
The explosion of variance is due to the phenomenon of \emph{weight degeneracy} (see Doucet \& Johansen (2011) and the references therein), which is essentially that the variance of the product term in  \eqref{eq:like_is_exp} generally increases as the number of product terms increases.
This issue is clearly undesirable for any 
problem, even when $t$ is small; see Whiteley et al.~(2012) for some discussion.

The representation \eqref{eq:like_vertex} also allows one to consider a single importance distribution on the space of permutations, associated to the removal of vertices,
as adopted in Guetz \& Holmes (2011). The latter authors identify a clever proposal based upon models on permutations. However, such a technique is likely to fail as $t-t_0$ grows; such
importance sampling methods are often subject to the curse of dimensionality and have a cost of $\mathcal{O}(\kappa^{t-t_0})$ for some $\kappa>1$ - see Bickel et al.~(2008). One idea
that can get around this problem is the work in Del Moral et al.~(2006) (see Beskos et al.~(2011)) as implemented by Guetz \& Holmes (2011), but needs to store $N$ samples of $t-t_0$ vertices in parallel; this is more expensive than the methods to be discussed below.

\subsection{Sequential Monte Carlo}\label{sec:smc}



A potential way to deal with the exponential order of the relative variance of IS, as used
for example for hidden Markov models Doucet \& Johansen (2011), 
is to adopt SMC methods.  This procedure simulates a collection of particles in parallel,
sequentially, and the particles evolve via sampling and resampling. The algorithm is designed
to approximate a sequence of probabilities of increasing dimension.

Define
$$
w_{k}(\nu_{t-k:t}) = \frac{\mathbb{I}_{\mathsf{W}_{k}(\nu_{t-k+1:t})}(\nu_{t-k})\omega_{\theta}(\delta^k(G_t,\nu_{t-k+1:t}),\nu_{t-k})}{(t-k)q_{\theta_0}(\nu_{t-k}|\nu_{t-k+1:t})},
$$
then, consider the algorithm:
\begin{enumerate}
\item{For $i\in\{1,\dots,N\}$ sample $\nu_t^i$ from $q_{\theta_0}(\cdot)$ and calculate $w_{0}(\nu_{t}^i)$. Set $k=0$.}
\item{Normalize the weights $\bar{w}_{k}^i = w_{k}(\nu_{t-k:t}^i)/\sum_{j=1}^N w_{k}(\nu_{t-k:t}^j)$ and resample the particles; denote them
$(\tilde{\nu}_{t-k:t}^i)_{1\leq i \leq N}$. Set $k=k+1$; if $k=t-t_0$ stop.}
\item{For $i\in\{1,\dots,N\}$ sample $\nu_{t-k}^i|\tilde{\nu}_{t-k+1:t}^i$ from $q_{\theta_0}(\cdot|\tilde{\nu}_{t-k+1:t}^i)$ and calculate $w_{k}(\nu_{t-k}^i,\tilde{\nu}_{t-k+1:t}^i)$.
Denote the particles $(\nu_{t-k:t}^i)_{1\leq i \leq N}$ and return to 2.}
\end{enumerate}

The estimate of the likelihood (computed after steps 1.~and 3.) is
\begin{equation}
L_{\theta}^N(G_t) = \prod_{k=0}^{t-t_0-1} \bigg[\frac{1}{N}\sum_{i=1}^N w_{k}(\nu_{t-k:t}^i)\bigg]
\label{eq:smc_like}
\end{equation}
and is unbiased for any fixed $N$ (Del Moral, 2004). 
We remark that the algorithm provides consistent estimates as $N$ grows; see Del Moral (2004) 
for details. We have the following result, whose proof is given in the appendix. The notations and assumptions are also fully described in the appendix. The expectation below, is w.r.t.~the process associated to SMC 
algorithm just described.

\begin{prop}\label{prop:rel_var}
Assume (A\ref{assumption}). Then if $N> \xi(t-t_0) \sum_{k=0}^{t-t_0-1} \xi_k(t-t_0) $, we have
$$ 
\mathbb{E}\bigg[\Big(\frac{L_{\theta}^N(G_t)}{L_{\theta}(G_t)}-1\Big)^2\bigg] \leq  \frac{4\xi(t-t_0)}{N}\sum_{k=0}^{t-t_0-1} \xi_k(t-t_0)
$$
where $L_{\theta}^N(G_t)$ is as \eqref{eq:smc_like}.
\end{prop}

\begin{rem}\label{rem:rel_var}
The constants $\xi(t-t_0),\xi_k(t-t_0)$ depend upon the number of removable nodes. In the case that $q(\cdot|\nu_{t-k+1:t})$ is uniform on the number of removable nodes, one can show that $\xi_k(t-t_0)=t-t_0-k$ and that $\xi(t-t_0)\leq t-t_0$;
so if $N>(t-t_0)^3$, then the relative variance of the likelihood estimate grows at most as $(t-t_0)^3$. This is opposed to the exponential order for IS. In general, one does not
expect the relative variance to grow linearly, as is the case for many other models; we explain this below.
\end{rem}

\subsubsection{Some Remarks}

The SMC algorithm targets the sequence of probabilities:
$$
\pi_k(\nu_{t-k:t}) \propto
\prod_{j=0}^{k}\bigg\{
\frac{\mathbb{I}_{\mathsf{W}_{j}(\nu_{t-j+1:t})}(\nu_{t-j})\omega_{\theta}(\delta^k(G_t,\nu_{t-j+1:t}),\nu_{t-j})}{(t-j)}
\bigg\}.
$$
If one uses the optimal proposal \eqref{eq:opt_proposal}, then importance weights are proportional to
$$
\frac{\sum_{\nu_{t-k}}\pi_k(\nu_{t-k:t})}{\pi_k(\nu_{t-k+1:t})}.
$$
The algorithm may not perform well if there is a significant discrepancy between these two probabilities. In practice one does not expect this to be the case, but this issue could be dealt with by
the approach in Doucet et al.~(2006). We note that at time $k$, one does not need to store $N$ trajectories of length $k$ (the reasons for which are rather technical); see Jacob et al.~(2013) for details.

The resampling mechanism in point 2., is the operation of sampling the particles with replacement according the current collection of weights. There are a wide variety of techniques to perform resampling and we will use the stratified resampling approach; see Doucet \& Johansen (2011) and the references therein for details. Resampling generally deals quite well with the weight degeneracy problem, but induces another problem, called \emph{path degeneracy}.
When resampling at every time step, at given time $k$ reasonably large (relative to $t-t_0$) the nodes that have been removed at the start (i.e.~say $\nu_{t-s:t}$, for some $s$)
are almost the same for each particle; this is because one never changes these values and this is the path-degeneracy problem (see Doucet \& Johansen (2011) for further details). 
This is why we do not expect that the result in Proposition \ref{prop:rel_var} to lead to a linear growth of the relative variance.
However, as the algorithm evolves on a finite state-space the variance of the $w_k$
can be well-behaved.
 In this scenario, the dynamic resampling
approach (e.g.~Del Moral et al.~(2012)) will partially alleviate the path degeneracy issue, which should not be overly troublesome for this problem. Here one will only resample when the weights are sufficiently variable; one way to meaure this is via the effective sample size (ESS):
$$
\frac{1}{\sum_{i=1}^N\bar{w}_{k}^i}.
$$
This is a number between $1$ and $N$ and generally resampling occurs if $\textrm{ESS}<N/2$. The ESS is simply an indication of the perfomance of the algorithm
and is not a fool-proof measure. For example, all of the samples may be in equally `bad' parts
of the state-space, leading to an ESS which is very high, but here the algorithm is not performing well. This issue will also manifest itself in the context of network models, because,
initially and close to reaching the irreducible network, many samples could be similar. As an additional performance indicator, we consider the number of unique particles, which can help
to determine how well the SMC algorithm is performing. Throughout the article, we use the SMC algorithm with dynamic resampling (unless otherwise stated), and we remark that the estimate of the normalizing constant is still unbiased. This estimate is as follows; suppose one resamples $s$ times at times $r_1,\dots,r_s$ and let $r_0+1=0$ and $r_{s+1}=t-t_0-1$, then
we have the estimate
$$
L_{\theta}^N(G_t) =\prod_{k=0}^s \frac{1}{N} \sum_{i=1}^N \prod_{j=r_{k}+1}^{r_{k+1}} w_{j}(\nu_{t-j:t}^i).
$$
As one expects the algorithm with dynamic resampling to perform better than resampling at each time-point; the relative variance of this estimate should not be worse than $\mathcal{O}((t-t_0)^3)$ as discussed above.


\subsection{Discrete Particle Filter} \label{sec:dpf}

As mentioned, the simulation (either IS or SMC) will evolve on a finite state-space.
The SMC method should deal with the exponential order of the relative variance of the likelihood estimate, but does not necessarily exploit the nature of the state-space; one may want to consider an alternative approach.
 In this scenario, one can use the discrete particle filter modified to the current
situation. The algorithm is now described.

\begin{enumerate}
\item{Set $S_0=\mathcal{R}(G_t)$, for each $\nu_t\in S_0$ compute
$$
w_0(\nu_t) = \frac{1}{t}\omega_{\theta}(G_t,\nu_t) \quad \bar{w}_0(\nu_t) =  \frac{w_0(\nu_t)}{\sum_{\nu_t\in S_0} w_0(\nu_t)}.
$$
}
\item{At times $1\leq k \leq t-t_0-1$
\begin{enumerate}
\item{If $\textrm{Card}(S_{k-1})\leq N$, set $S_{k-1}'=S_{k-1}$, $C_{k-1}=\infty$ and go to (b). Otherwise, perform the resampling step described below, which returns $S_{k-1}'$ and $C_{k-1}$.}
\item{Set $S_k=\{\nu_{t-k:t}:\nu_{t-k+1:t}\in S_{k-1}',\nu_{t-k}\in \mathsf{W}_k(\nu_{t-k+1:t}) \}$}
\item{For each $\nu_{t-k:t}\in S_k$ compute
$$
w_k(\nu_{t-k:t}) = \frac{1}{t-k}\omega_{\theta}(\delta^k(G_t,\nu_{t-k+1:t}),\nu_{t-k})\frac{w_{k-1}(\nu_{t-k+1:t})}{1\wedge C_{k-1} w_{k-1}(\nu_{t-k+1:t})} 
$$
$$
\bar{w}_k(\nu_{t-k:t}) =  \frac{ w_k(\nu_{t-k:t}) }{ \sum_{\nu_{t-k:t}\in S_k} w_k(\nu_{t-k:t}) }.
$$
}
\end{enumerate}
}
\end{enumerate}

In part 2 (a), we have the following procedure.
\begin{itemize}
\item{Set $C_{k-1}$ to be the unique solution of
$$
\sum_{\nu_{t-k+1:t}\in S_{k-1}} 1\wedge C_{k-1} w_{k-1}(\nu_{t-k+1:t}) = N.
$$
}
\item{Keep the $L_{k-1}$, $\nu_{t-k+1:t}$ whose weights are greater than $1/C_{k-1}$. For the remaining $\textrm{Card}(S_{k-1})-L_{k-1}$ particles perform the following stratified resampling scheme.}
\item{Normalize the weights $\bar{w}_{k-1}(\nu_{t-k+1:t})$ of the remaining $|S_{k-1}|-L_{k-1}$ and label them to obtain the normalized weights $\hat{w}_{k-1}(\nu_{t-k+1:t}^i)$, $i\in\{1,\dots,\textrm{Card}(S_{k-1})-L_{k-1}\}$.}
\item{Construct the CDF: for $i\in\{1,\dots,\textrm{Card}(S_{k-1})-L_{k-1}\}$
$$
Q_{k-1}(i) := \sum_{j=1}^i \hat{w}_{k-1}(\nu_{t-k+1:t}^j), \quad  Q_{k-1}(i) :=0.
$$
}
\item{Sample $U_1$ uniformly on $[0,1/(N-L_{k-1})]$ and set $U_j=U_1 + (j-1)/(N-L_{n-1})$, $j\in\{2,\dots,N-L_{k-1}\}$}
\item{For $i\in\{1,\dots,\textrm{Card}(S_{k-1})-L_{k-1}\}$, if there exist a $j\in\{1,\dots,N-L_{k-1}\}$ such that $Q_{k-1}(i-1)<U_j\leq Q_{k-1}(i)$, then $\nu_{t-k+1:t}^i$ survives.}
\item{Set $S_{k-1}'$ to be the set of surviving particles from the resampling and the $L_{k-1}$ samples that were maintained.}
\end{itemize}

The only stochasticity in the algorithm is brought about by the resampling mechanism. The finite state-space of the samples is exploited
by deterministically diversifying the particles, forcing them to explore each part of the state-space.
This algorithm differs slightly from the standard DPF in that the support for a given particle may be different than another.
It is exact, if $\textrm{Card}(S_{t-t_0-2})\leq N$, but this is unlikely to be possible in practice; as we noted earlier the exact calculation is worse than exponential order in $t-t_0$.
One can show, using the same reasoning as Whiteley et al.~(2010) that the estimate:
$$
L_{\theta}^N(G_t) = \prod_{k=0}^{t-t_0-1}  \sum_{\nu_{t-k:t}\in S_k } w_k(\nu_{t-k:t})
$$
is unbiased for any fixed $N$.
In general, for a large network, this algorithm may be very expensive to implement initially. However, it can be used, once the network is sufficiently small (e.g.~half way through an SMC algorithm).
The algorithm has been shown to be rather efficient relative to other methods in the context of switching state-space models. We expect that this method will work rather well for moderate size networks and perhaps much better than SMC; thus we do not analyze the relative variance of the estimate of the likelihood as this is expected to be at least as good as SMC.

\subsection{Parameter Estimation}\label{sec:pmcmc}

In the discussion above, we have reviewed and developed methods for likelihood estimation, for one fixed $\theta$, we will now show how these approaches can be used to perform Bayesian parameter estimation.
Let $p(\theta)$ be a proper prior density on $\Theta$, then Bayesian parameter inference is concerned with the posterior
$
\pi(\theta|G_t) \propto L_{\theta}(G_t) p(\theta).
$
As we have remarked, the direct calculation of $L_{\theta}(G_t)$ is not possible, so for example,
if one wanted to perform Monte Carlo estimation associated to the posterior (which is often the only amenable way to do inference) then one could not implement standard IS or MCMC algorithms, which will require the evaluation of $L_{\theta}(G_t)$.

A recently evolving MCMC methodology in computational statistics turns out to be rather useful for the class of problems of interest. These ideas are exact approximations of idealized MCMC algorithms; in our context, the ideal algorithm would be simply to sample from the posterior using e.g.~Metropolis-Hastings or the Gibbs sampler. Since either procedure would require us to evaluate $L_{\theta}(G_t)$, they are idealized algorithms. However, as noted in 
Andrieu \& Roberts (2009), one need only know the target probability up-to an un-normalized, unbiased estimate. That is, given some auxilliary variable $u\in\mathsf{U}$ with probability density $f_{\theta}(u)$, (with associated probability measure $F_{\theta}(\cdot)$) if one has, for any fixed $\theta\in\Theta$:
$$
\pi(\theta|G_t) \propto \Big(\int_{\mathsf{U}} L_{\theta}(G_t,u)f_{\theta}(u)du\Big) p(\theta)
$$
then one can construct an ergodic MCMC algorithm with target probability
$$
\pi(\theta,u|G_t) \propto L_{\theta}(G_t,u)f_{\theta}(u) p(\theta)
$$
and still obtain a Monte Carlo approximation of the posterior $\pi(\theta|G_t)$.

For each of the algorithms in Sections \ref{sec:smc} and \ref{sec:dpf}, we have remarked that the estimate of the likelihood is unbiased. Thus, for example taking the SMC algorithm and
denoting all the simulated random variables by $u$, one could sample from the target
$\pi(\theta,u|G_t) \propto L_{\theta}^N(G_t,u)f_{\theta}(u) p(\theta)$, where the auxilliary
variables are generated from the SMC algorithm and $L_{\theta}^N(G_t,u)$ is the associated likelihood estimate. This is the idea which has been developed in
Andrieu et al.~(2010) and Whiteley et al.~(2010) (for SMC and the DPF respectively).
The simplified particle marginal Metropolis-Hastings algorithm that we propose is as follows,
where $f_{\theta}(\cdot)$ could be associated to either the SMC algorithm in Section \ref{sec:smc} or
DPF algorithm in Section \ref{sec:dpf}. Below $q(\theta'|\theta)$ is a positive conditional probability density (with probability measure $Q(\cdot|\theta)$).

\begin{enumerate}
\item{Sample $\theta$ from the prior, $U$ from $F_{\theta}(\cdot)$ and compute the likelihood estimate $L_{\theta}^N(G_t,u)$. Set $\theta^0=\theta$ and $i=1$.}
\item{Sample $\theta'|\theta^{i-1}\sim Q(\cdot|\theta^{i-1})$ and $U$ from
$F_{\theta'}(\cdot)$ and compute the likelihood estimate $L_{\theta'}^N(G_t,u')$. Set $\theta^i=\theta'$ with probability:
$$
1\wedge \frac{L_{\theta'}^N(G_t,u')p(\theta')}{L_{\theta^{i-1}}^N(G_t,u)p(\theta^{i-1})}\times
\frac{q(\theta^{i-1}|\theta')}{q(\theta'|\theta^{i-1})}
$$
otherwise set $\theta^i=\theta^{i-1}$. Set $i=i+1$ and return to the start of 2.
}
\end{enumerate}
Under mild conditions, the above algorithm is ergodic; see Andrieu et al.~(2010) or Whiteley et al.~(2010) for details.
One major point is that $N$ is a user-set parameter and this is intrinsically linked to the (relative) variance of the likelihood estimator, which
is important in determining the (good) performance of this algorithm; see Andrieu et al.~(2010).
For network models, according to Proposition \ref{prop:rel_var}, one should take  $N=\mathcal{O}((t-t_0)^3)$ (see Remark \ref{rem:rel_var}). 
See also Doucet et al.~(2012) for results on choosing $N$.

\section{Numerical Illustrations}\label{sec:numerics}

Throughout this Section all code was written in MATLAB.

\subsection{Simulation Results for a Small Network} \label{sec:likelihood sim}

In this Section, we will investigate the methods in Section \ref{sec:comp} and in particular in terms of their accuracy along with some of the phenomona mentioned above. With regards to the data, since the true likelihood is only computable when the size of the network is small, for our purpose, we use the graph in Fig.~1 in Wiuf et al.~(2006). It is an example of a graph generated with parameter $\theta=(1,0.66,0.33,0)$ under the DA model. It has 10 nodes where most nodes are removable, at the end, this graph can be reduced to a single node. We have five subsections in this part, Sections \ref{sim:is}-\ref{sim:dpf} simply display results of the IS method, the SMC method and the DPF method respectively; Section \ref{sim:rv} deals with the relative variance between the true likelihood and each of estimates obtained by the above three methods; then Section \ref{sim:cpu} compares the results of IS, SMC and DPF method obtained in approximately the same computing time. Section \ref{sim:pmcmc} investigates the PMCMC algorithms discussed in Section \ref{sec:pmcmc}.

\subsubsection{Importance Sampling}\label{sim:is}
We apply the IS method of Section \ref{sec:is} to our network model. Here and in the following subsections, we fix three parameters in $\theta=(\pi, p, q, r)$ except $p$ to the true values, i.e., $\pi=1$, $q=0.33$, $r=0$, and $p\in \{0.05, 0.15, \ldots, 0.85\}$. We set the driving value $\theta_0=(1,0.66,0.33,0)$ except in Section \ref{sim:cpu}, where we set $\theta_0=\theta$. Then we estimate the likelihood under these different nine $\theta$ to obtain the estimated likelihood curve with respect to parameter $p$. We run IS 30 times with both $N\in\{1000,10000\}$ respectively, then use the average to be the estimator, and construct a confidence interval. In order to demonstrate the issue about the relative variance of the IS method, we calculate the value of ESS at the end of a single run for each parameter $p\in\{0.05, 0.15, \ldots, 0.85\} $, with $N\in\{100, 1000, 10000\}$. The results are given in Figure \ref{fig:is}.  

From Figure \ref{fig:is}, plots (a) and (b) display the similarity of these two estimated curves: they both approximate the true likelihood curve well, but not exactly. The main difference in
the plots are the expected improved accuracy and reduced variance as $N$ grows.
Plot (c) shows that for each $N\in\{100,1000,1000\}$, the values of the $ESS$ are quite low. As a result, the variance of the un-normalized weights is quite large; this occurs due to the weight degeneracy problem.

\begin{figure}[!tpb]
\centering
\subfigure[$N=1000$]
{\includegraphics[width=15.5cm,height=4cm]{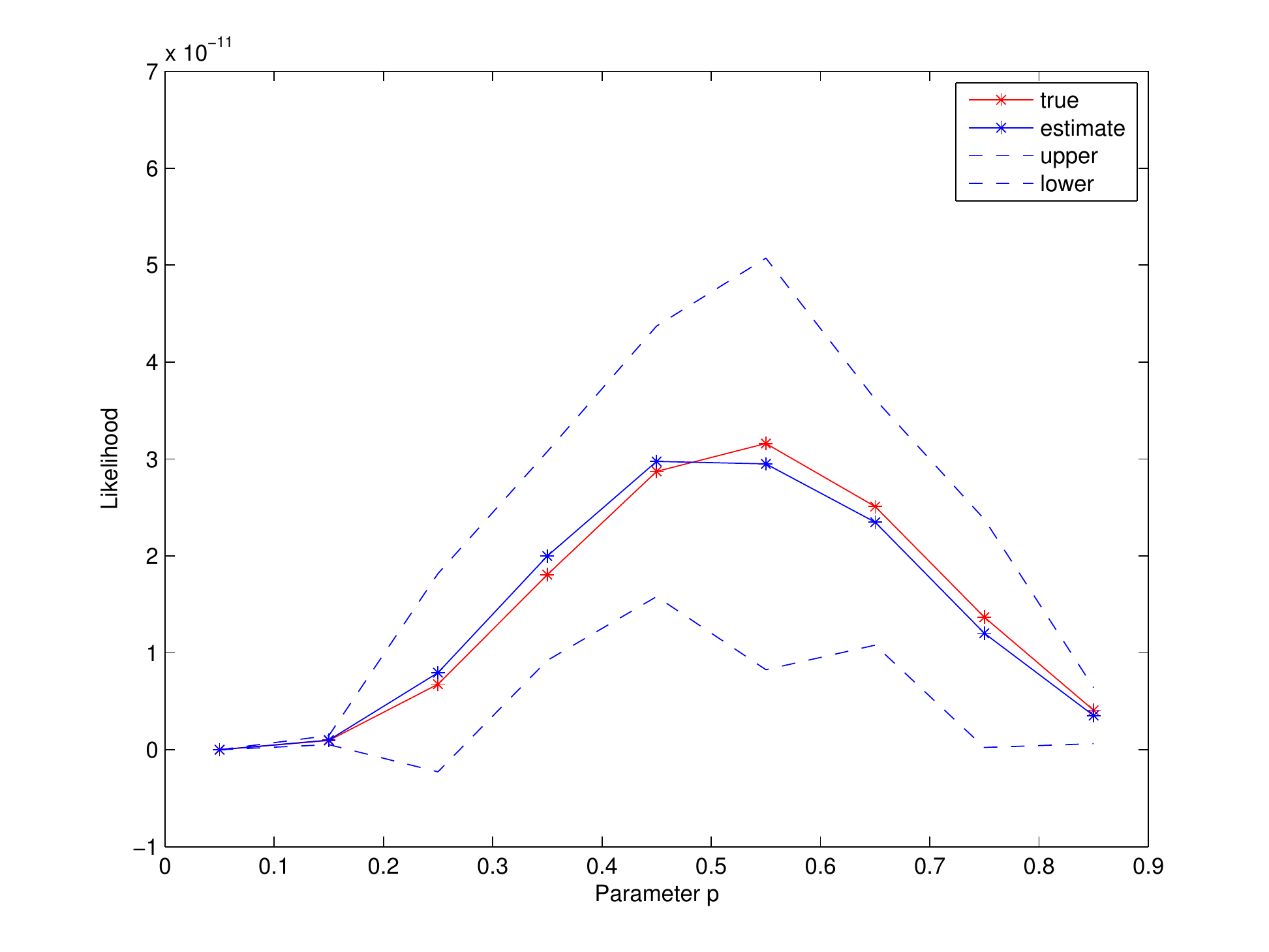}}
\subfigure[$N=10000$]
{\includegraphics[width=15.5cm,height=4cm]{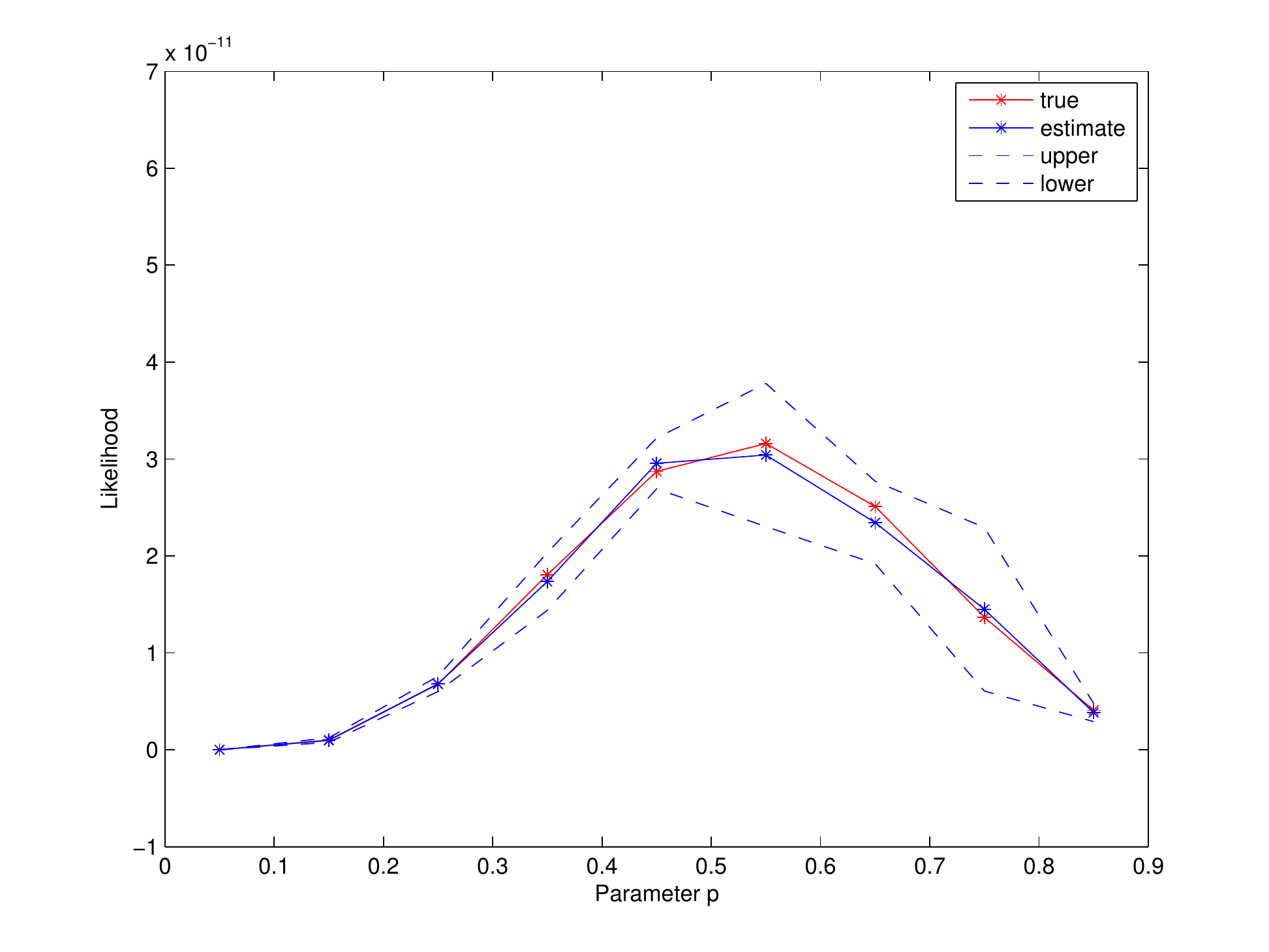}}
\subfigure[ESS]
{\includegraphics[width=15.5cm,height=8cm]{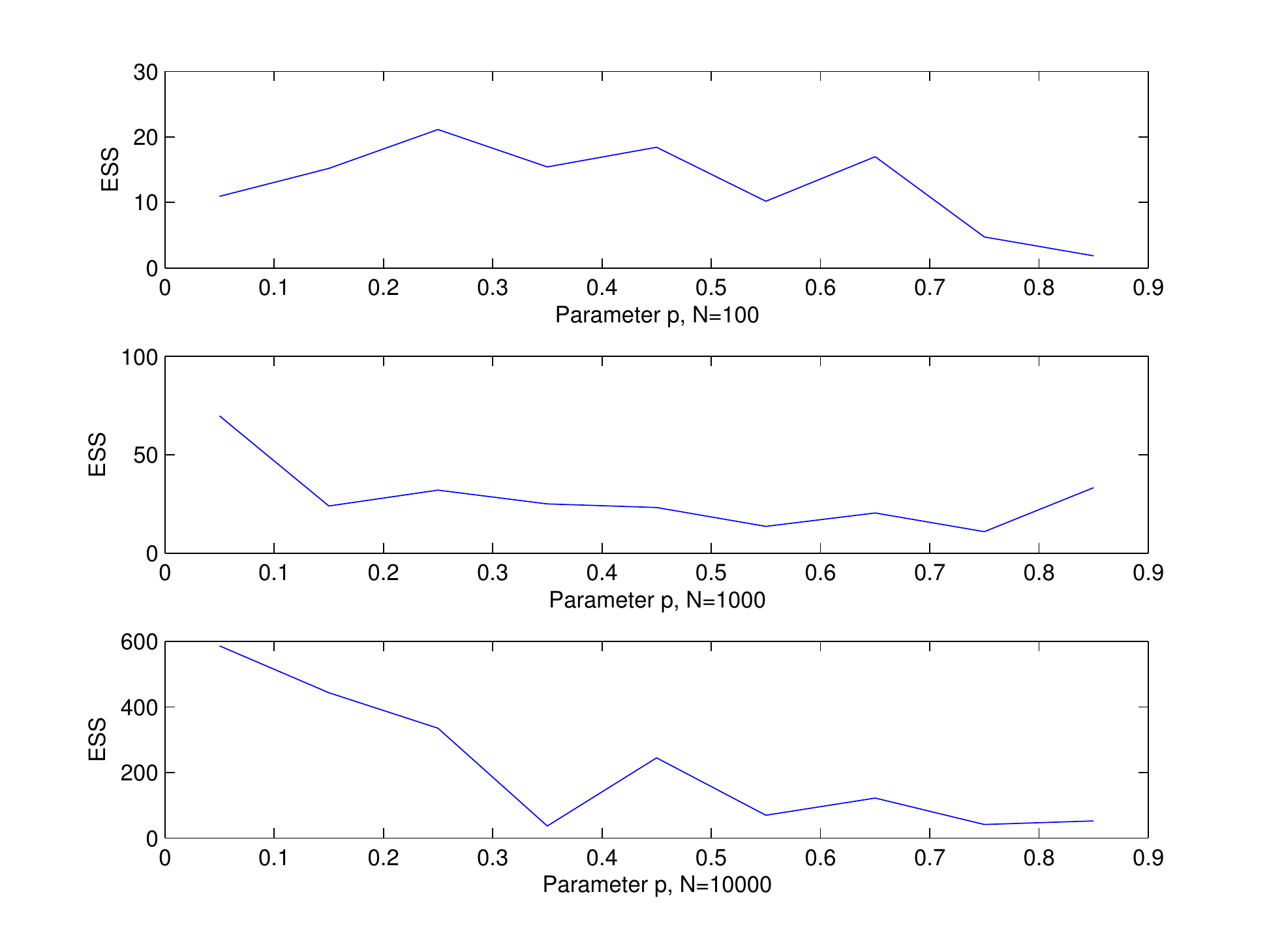}}
\caption{ \small Simulation results of IS algorithm: figures (a) and (b) are plots of estimated likelihood curve of 30 runs under $N=1000$ and $N=10000$ respectively, the red solid line with stars is the true likelihood the blue solid line with stars is the mean of 30 estimates, and the other two blue dashed lines are $\bar{x}-2s$ and $\bar{x}+2s$ ($s$ is the standard deviation accross the runs) respectively; figure (c) are plots of $ESS$ at the end of a single run for each $p$ under $N=100$, $N=1000$ and $N=10000$ (from upper to bottom).}  
\label{fig:is}
\end{figure}

\subsubsection{Sequential Monte Carlo Method}\label{sim:smc}

We now consider the SMC method to deal with the weight degeneracy. Similarly to the previous subsection, we set $\pi=1$, $q=0.33$, $r=0$, and let $p\in\{0.05, 0.15, \ldots, 0.85\} $. In order to keep consistency and potentially maximize the convergence rate, we also set the driving value $\theta_0$ equal to $(1,0.66,0.33,0)$. After some experiments, we find that the stratified resampling scheme (see e.g.~Doucet \& Johansen (2011)) outperforms other resampling schemes which is why it is adopted. We consider two resampling schemes of both dynamically resampling and resampling at each time, we run SMC 30 times with $N\in\{1000,10000\}$ and also construct a confidence interval via the repeats. 
The results are in Figures \ref{fig:smceverytime} and \ref{fig:smcdynamic}.

From plots (a) and (b) in Figures \ref{fig:smceverytime} and \ref{fig:smcdynamic} the performance of the estimation of the likelihood curve improves
as $N$ grows, as one would expect. It is also clear that dynamically resampling is more
accurate; again this is expected as one does not resample `too often' removing promising particles and exacerbating the path degeneracy problem.
To investigate the dynamically resampling SMC algorithm, we consider the ESS and the number
of unique particles, for a typical run of the algorithm, with $N\in\{100,1000,1000\}$ in Figure 
\ref{fig:smcdynamic}. 


For plot (c) of Figure \ref{fig:smcdynamic}, firstly, for the results of $UN$, with $N\in\{100, 1000, 10000\}$, the value of UN tends to approach the sample size at the last time. This means that SMC method works very well in this example, there is little path degeneracy, so that the results of the SMC are reliable. Note at the first few times, the value of UN is quite low since we have to choose 100, 1000, or 10000 at time one from 10, or 9, or an even smaller number of nodes. Secondly, for the results of ESS, almost at every time, the ESS is above the resampling threshold of $N/2$. This shows that the stratified resampling (dynamically) SMC can help to deal with the variance of the weights, which was evident in the IS example. This also means that this approach 
can help to deal with the relative variance issue encountered by IS,  and this will be verified in section \ref{sim:rv}. The results shown in this figure are quite satisfactory, as one would hope for in such a simple example.


\begin{figure}[!h]
\centering
\subfigure[$N=1000$, Resampling each time.]
{{\includegraphics[width=6.5cm,height=4cm]{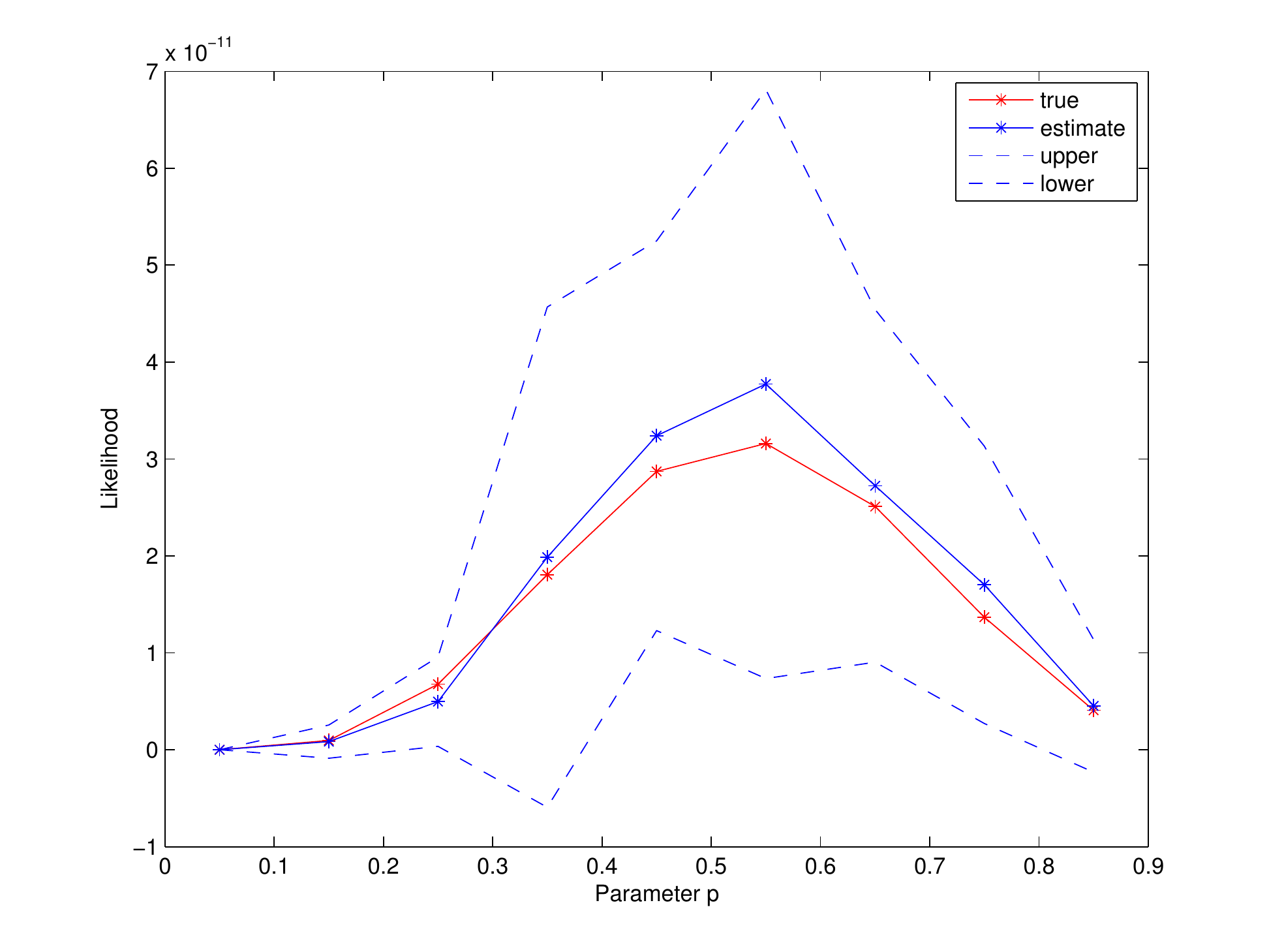}}}
\subfigure[$N=10000$, Resampling each time.]
{{\includegraphics[width=6.5cm,height=4cm]{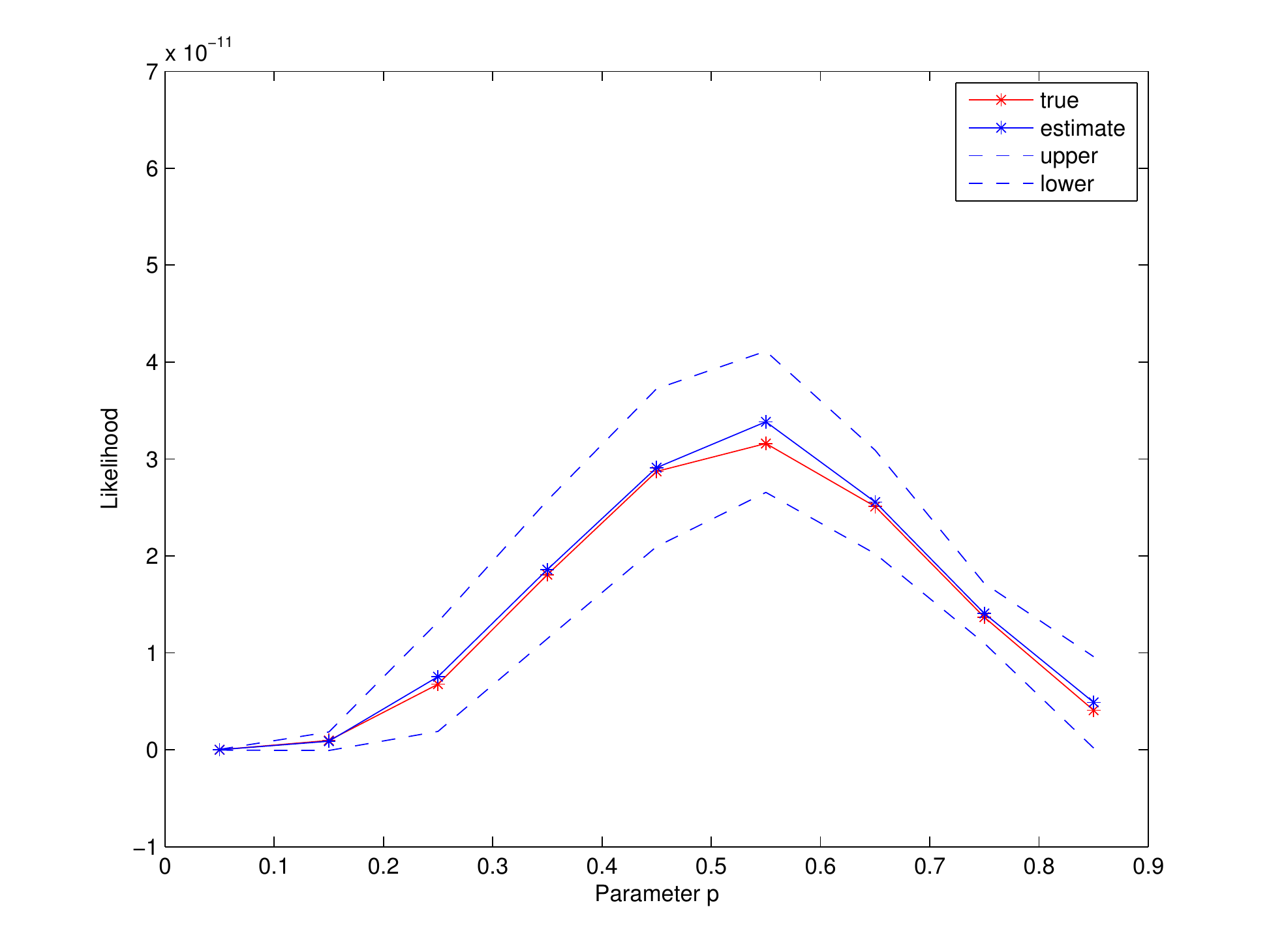}}}
\subfigure[$N=1000$]
{\includegraphics[width=6.5cm,height=4cm]{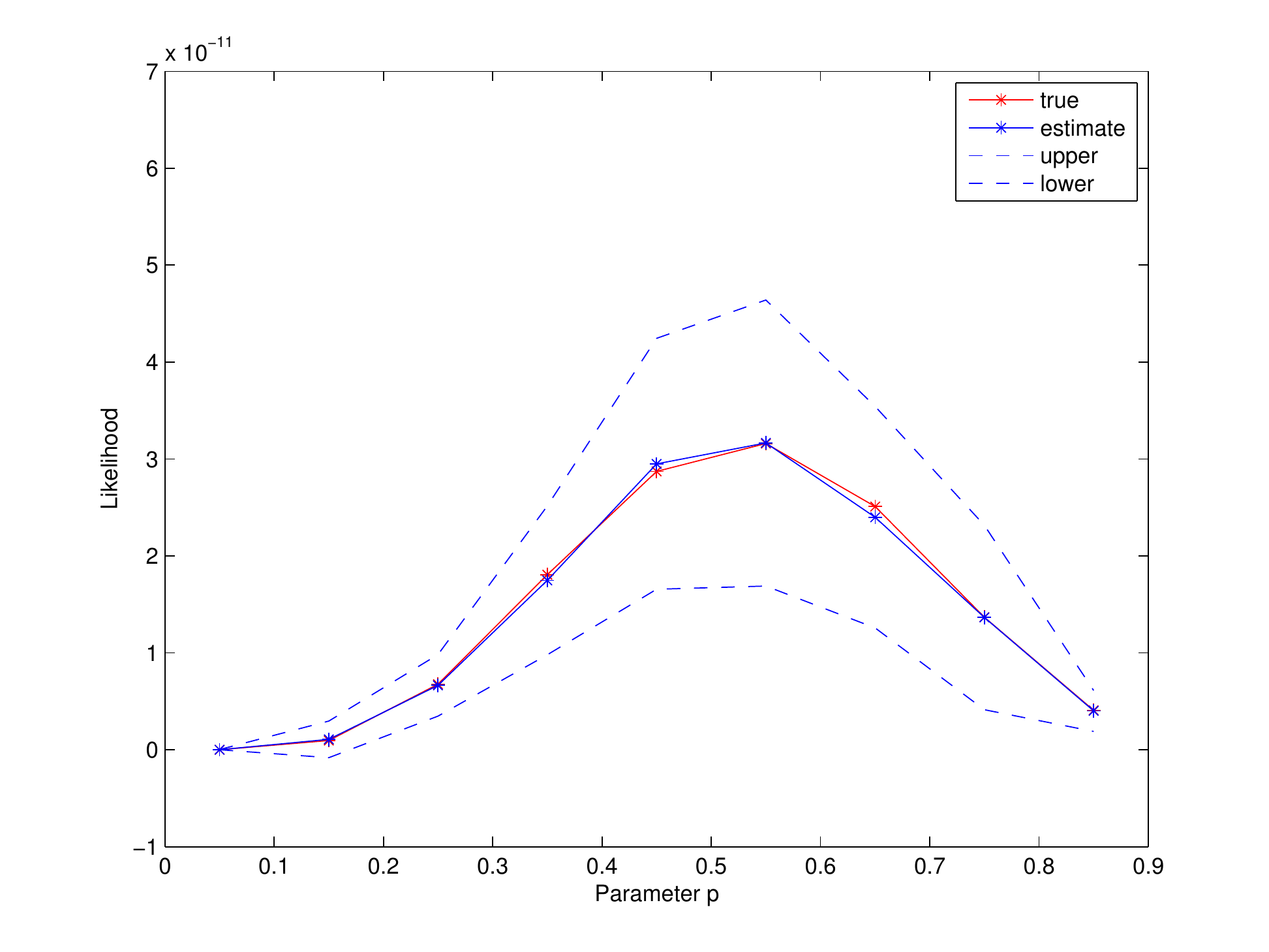}}
\subfigure[$N=10000$]
{\includegraphics[width=6.5cm,height=4cm]{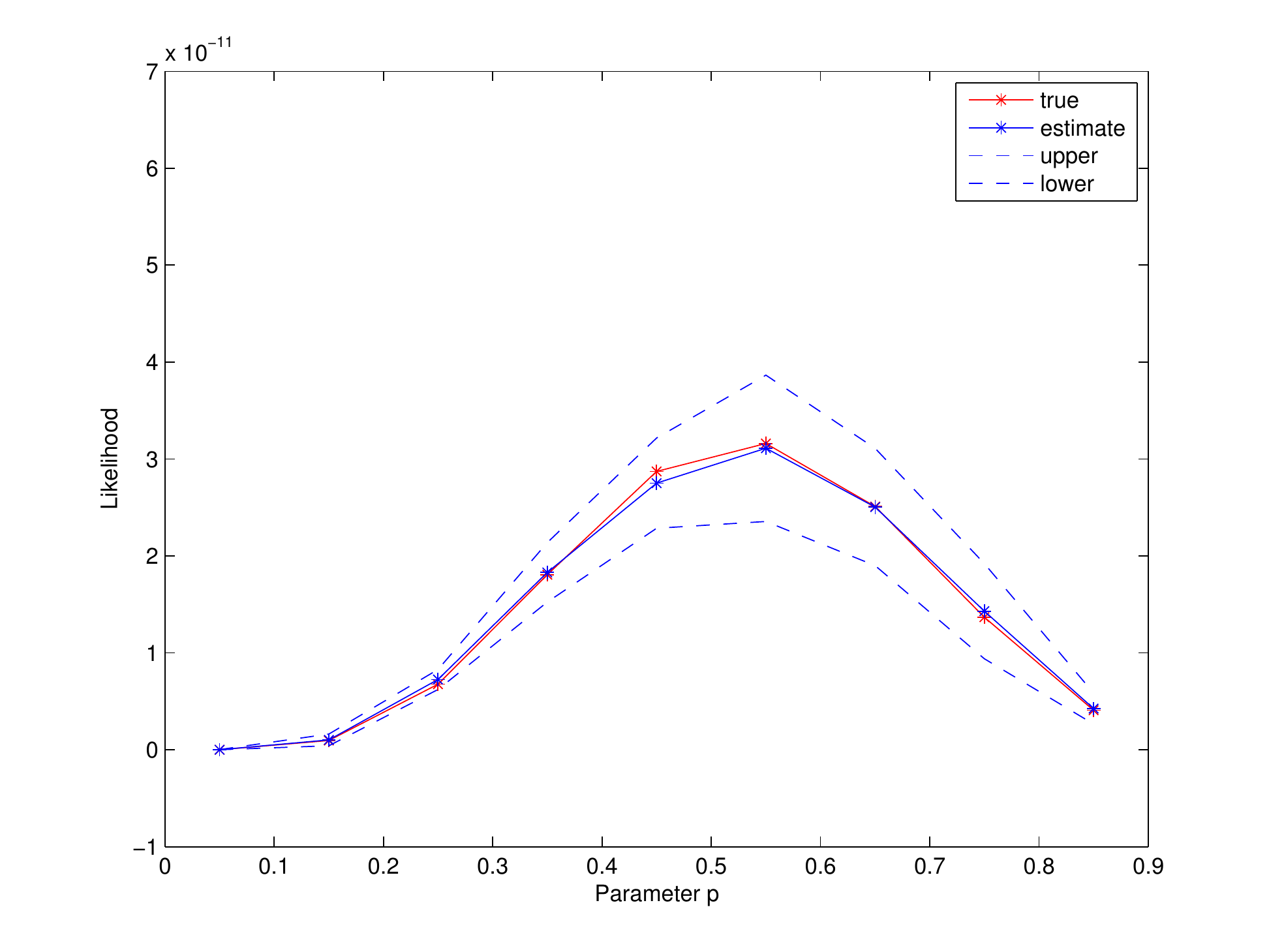}}
\caption{ \small Simulation results of the stratified resampling at every time SMC algorithm and dynamically resampling algorithm (bottom row). The estimates are the average across 30 runs
 with the upper and lower lines the $\pm$ 2 standard deviations, across the runs.}  
\label{fig:smceverytime}
\end{figure}

\begin{figure}[!tpb]
\centering
\subfigure[$ESS$ and $UN$.]
{\includegraphics[width=15.5cm,height=6cm]{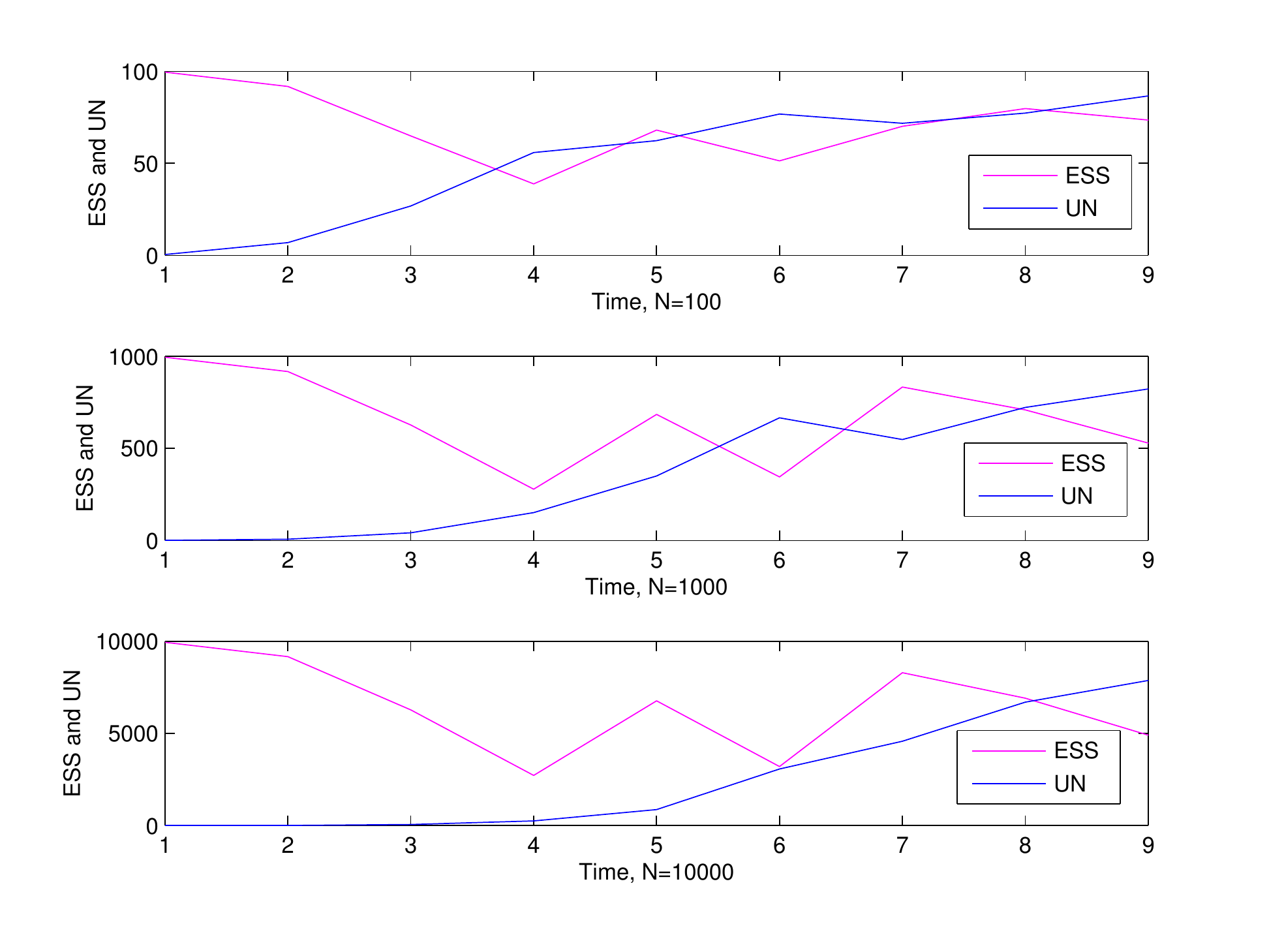}}
\caption{ \small Simulation results of the stratified resampling dynamically SMC algorithm: plots of the average of 30 runs of $ESS$ and $UN$ at every time with $\theta=(1,0.55,0.33,0)$ and $\theta_0=(1,0.66,0.33,0)$, under $N=100$, $N=1000$ and $N=10000$ (from upper to bottom).}  
\label{fig:smcdynamic}
\end{figure}

%

\subsubsection{ The Discrete Particle Filter Method}\label{sim:dpf}
We begin by noting that the DPF method does not need a proposal density, whilst in the previous IS and SMC sections, the proposal density is needed and represented by the driving value. For a fair comparison, in the latter section \ref{sim:cpu}, we will give simulation results of IS and SMC with dynamical stratified resampling when $\theta=\theta_0$.
 In this subsection, we only need to set $\pi=1$, $q=0.33$, $r=0$, and $p\in\{0.05, 0.15, \ldots, 0.85\} $. Then we run the DPF algorithm 30 times with $N\in\{100,1000,10000\}$ respectively to obtain the estimated likelihood curve and the confidence interval. Results are shown in Figure \ref{fig:dpf}.

From the three plots in Figure \ref{fig:dpf} as $N$ grows larger, the estimated likelihood curve tends to be closer to the true likelihood and the confidence interval narrower (versus the other methods). While the accuracy of the estimate is quite satisfactory, the estimated likelihood curve when $N=1000$ is almost the same as the true likelihood curve, and the confidence interval is very narrow. Note that the exact calculation is $\mathcal{O}((t-t_0)2^{t-t_0})$ (which is of order of about 4600) so one has almost the exact calculation when $N=10000$.
This is better than the previously mentioned methods, but we are yet to account for the computational time element.  

\begin{figure}[!tpb]
\centering
\subfigure[$N=100$]
{\includegraphics[width=15.5cm,height=4cm]{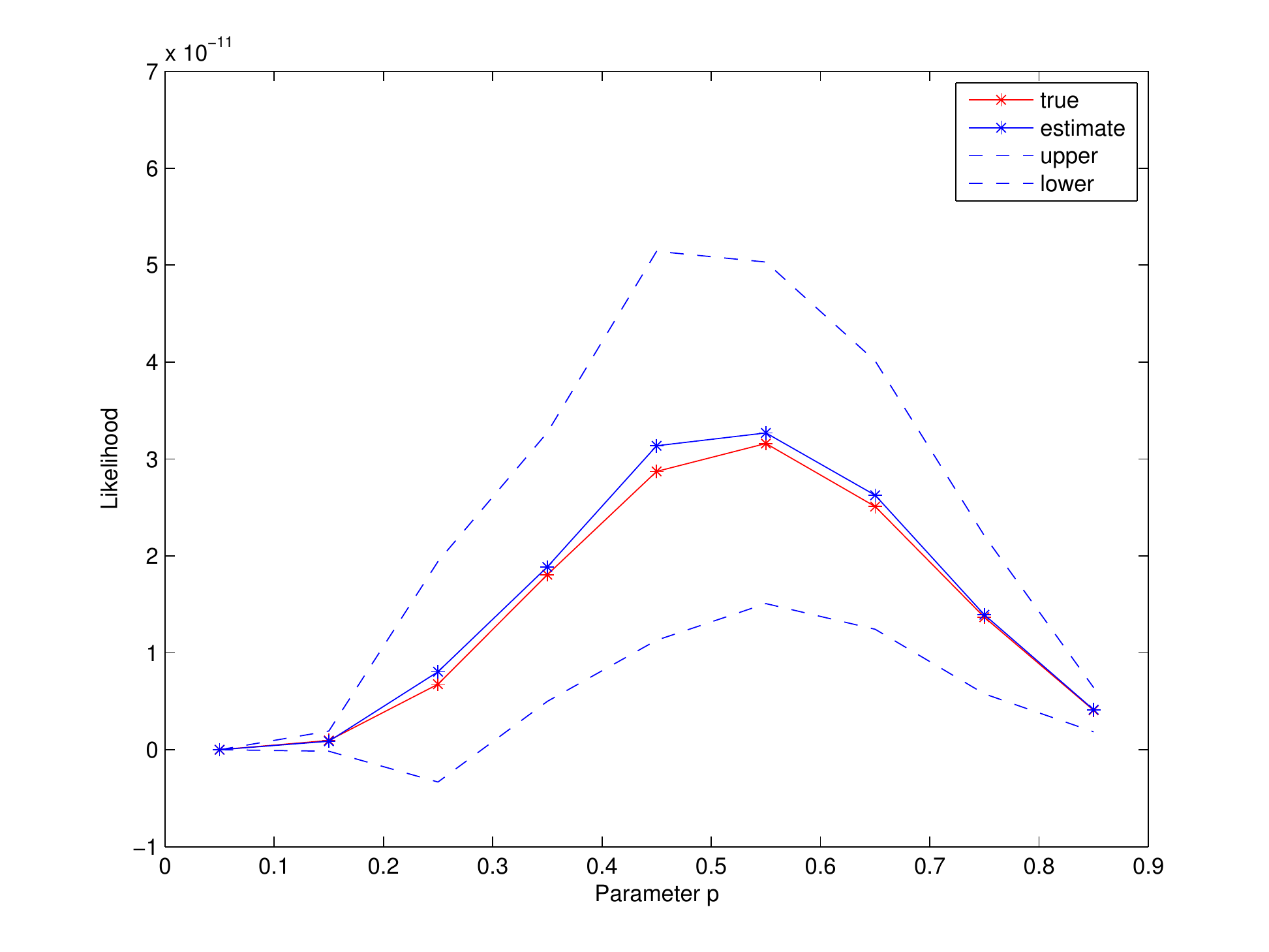}}
\subfigure[$N=1000$]
{\includegraphics[width=15.5cm,height=4cm]{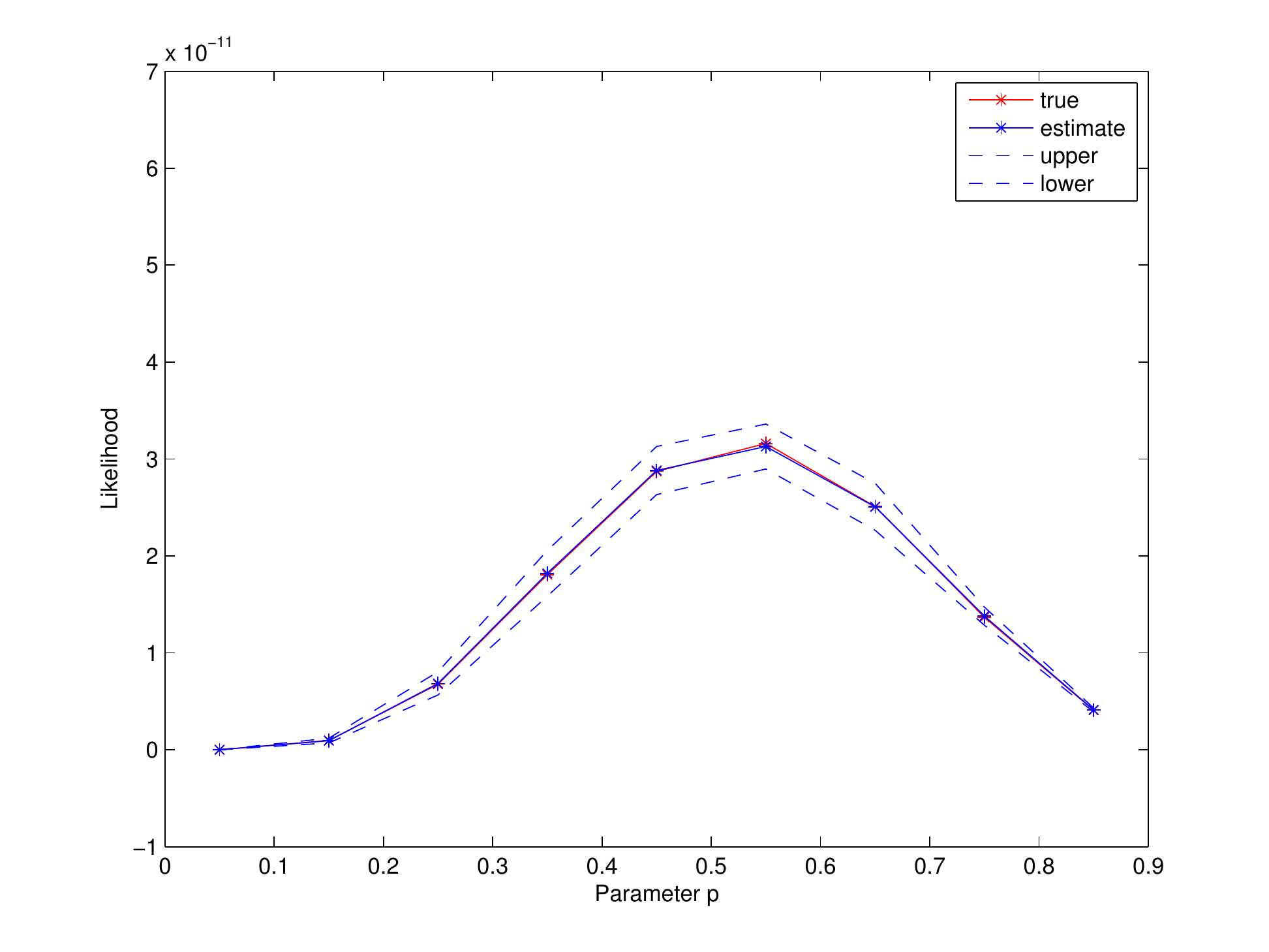}}
\subfigure[$N=10000$]
{\includegraphics[width=15.5cm,height=4cm]{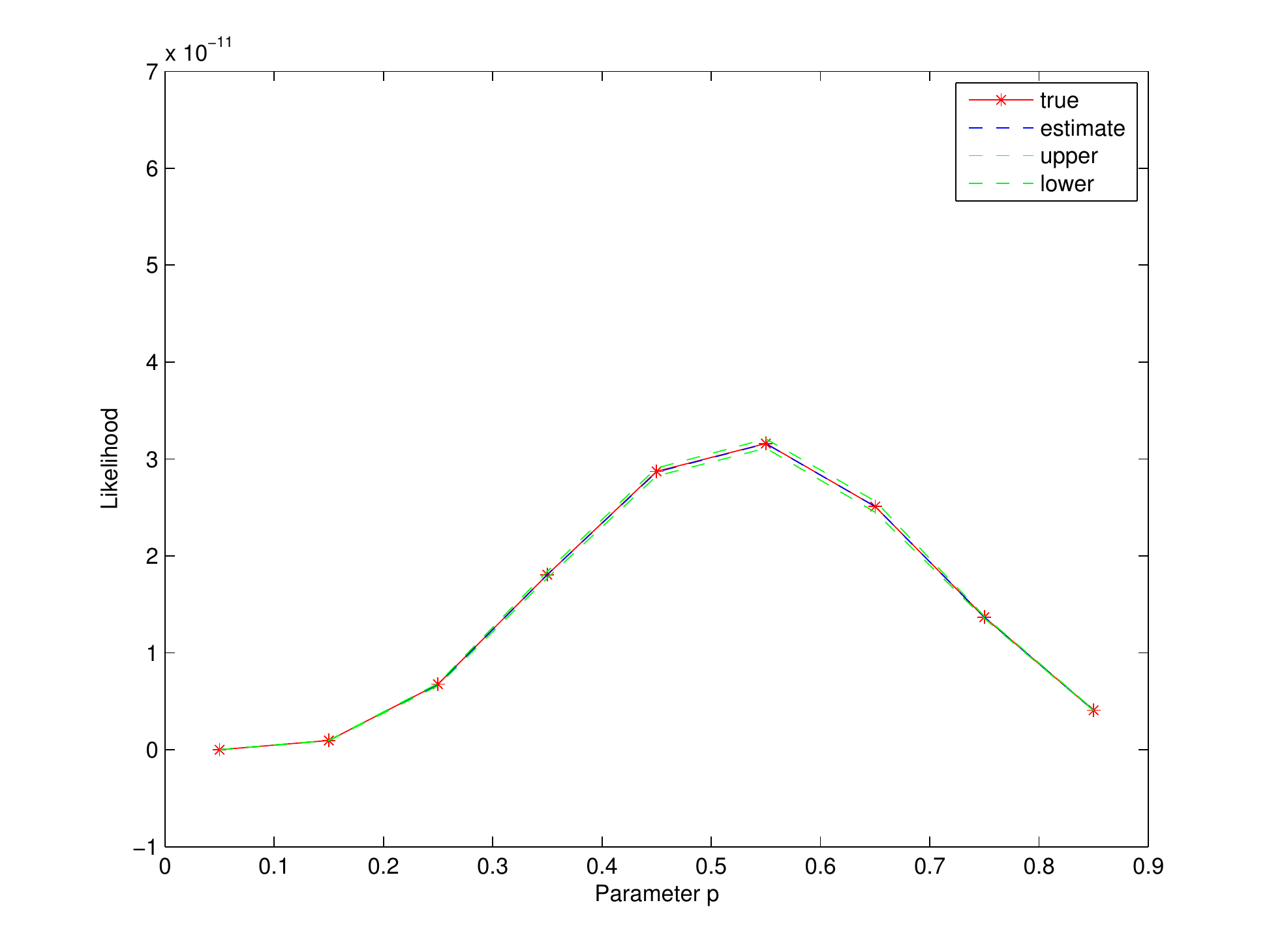}}
\caption{ \small Simulation results of DPF algorithm: figures (a)-(c) are plots of estimated likelihood curve of 30 runs under $N=100, 1000, 10000$ respectively, the red solid line with stars is the true likelihood the blue solid (or dashed) line with stars is the mean of 30 estimates, and the other two blue (or green) dashed lines are $\bar{x}-2s$ and $\bar{x}+2s$ respectively.}
\label{fig:dpf}
\end{figure}

\subsubsection{Relative Variance}\label{sim:rv}
In this subsection, we consider the relative variances
of the likelihood estimate, for each of the methods implemented above.
We focus on comparing relative variances as the size of the network changes. We use the DA model to generate different sizes of network models, from size 5 up to size 13. Then for each network model, we use the Monte Carlo methods to estimate the relative variance using $30$ repeats and 1000 particles with $\theta=(1,0.55,0.33,0)$, $\theta_0=(1,0.66,0.33,0)$ to obtain estimators. We remark that we can also obtain the true likelihood. The results are displayed in Table \ref{table:rv} and the conclusions seem to be consistent across different parameter values.

 From this table, basically, for all of these three methods, as the size of network model grows, the relative variance between the estimated likelihood and the true likelihood goes larger;  this is unsurprising as this is to be expected. 
Amongst these three methods, the DPF method has the smallest value of the relative variance and is significantly more so.
This tells us that the DPF algorithm provides us more accurate and reproducible estimators, at least for this example. As for the other two methods, for the size of network model below 10, there are tiny differences, for the size of network model beyond 10,  the stratified dynamically resampling SMC gives better results than the IS method. This latter result is consistent with the remarks in Section \ref{sec:is} and Section \ref{sec:smc} and also the results in Section \ref{sim:is} and Section \ref{sim:smc}. 

\begin{table}{}
\begin{center}
\begin{tabular}  {c|c|c|c}
SIZE & IS  & STRA &  DPF\\ 
\hline
5 &     0.0003 &  0.0002 & 0.0000 \\
6 &     0.0027  & 0.0030  &0.0000 \\
7 &   0.0043 & 0.0064  &0.0000 \\
8 & 0.0158 &  0.0142   &0.0000 \\
9 &  0.0149& 0.0136  &0.0010 \\
10 &  0.0419& 0.0128  & 0.0036 \\
11 & 0.1512& 0.0364  &0.0084 \\
12 &   0.5659& 0.1115  &0.0079 \\
13&  1.4224&0.3022   &0.0657 \\
\end{tabular}
\end{center}
\caption{\small Relative variance of the estimates of the above three methods w.r.t the true likelihood: here the results refer to the size of network from 5 up to 13, with $\theta=(1,0.55,0.33,0)$, $\theta_0=(1,0.66,0.33,0)$ and $N=1000$.} \label{table:rv}

\end{table}

\subsubsection{CPU Time}\label{sim:cpu}
Here we are mainly interested in comparing results of each of the methods obtained in approximately the same (wall-clock) computation time. As we mentioned in the previous subsection that the DPF method does not use a driving value so in order to obtain fair results 
we set $\theta=\theta_0$ for the IS and SMC approaches.

After a few experiments, we find that when we set $N=1000$ for both IS and DPF algorithms, and set $N=550$ for the SMC algorithm, then run 30 times for all of them, the run times are 1274.67 seconds (IS), 1289.31 seconds(SMC) and 1272.13 seconds (DPF). The estimated likelihood curves and confidence intervals are displayed in Figure \ref{fig:ct}.
 
In Figure \ref{fig:ct}, with regards to the estimated likelihood curves, the estimated likelihood curve obtained by the DPF method is closest to the true likelihood curve, and the estimated likelihood curve under the SMC method seems to a little closer to the true likelihood curve than the estimated likelihood curve under the IS method. The confidence interval under the DPF method is the narrowest,  and the confidence interval under the SMC method is also a little narrower than the confidence interval for the IS method. Thus, we can conclude that for the given computation time and this specific example, the DPF method can provide us the best estimator, and the SMC method is the second, then the IS method is the last.

\begin{figure}[!tpb]
\centering
\centerline{\includegraphics[width=15.5cm,height=6cm]{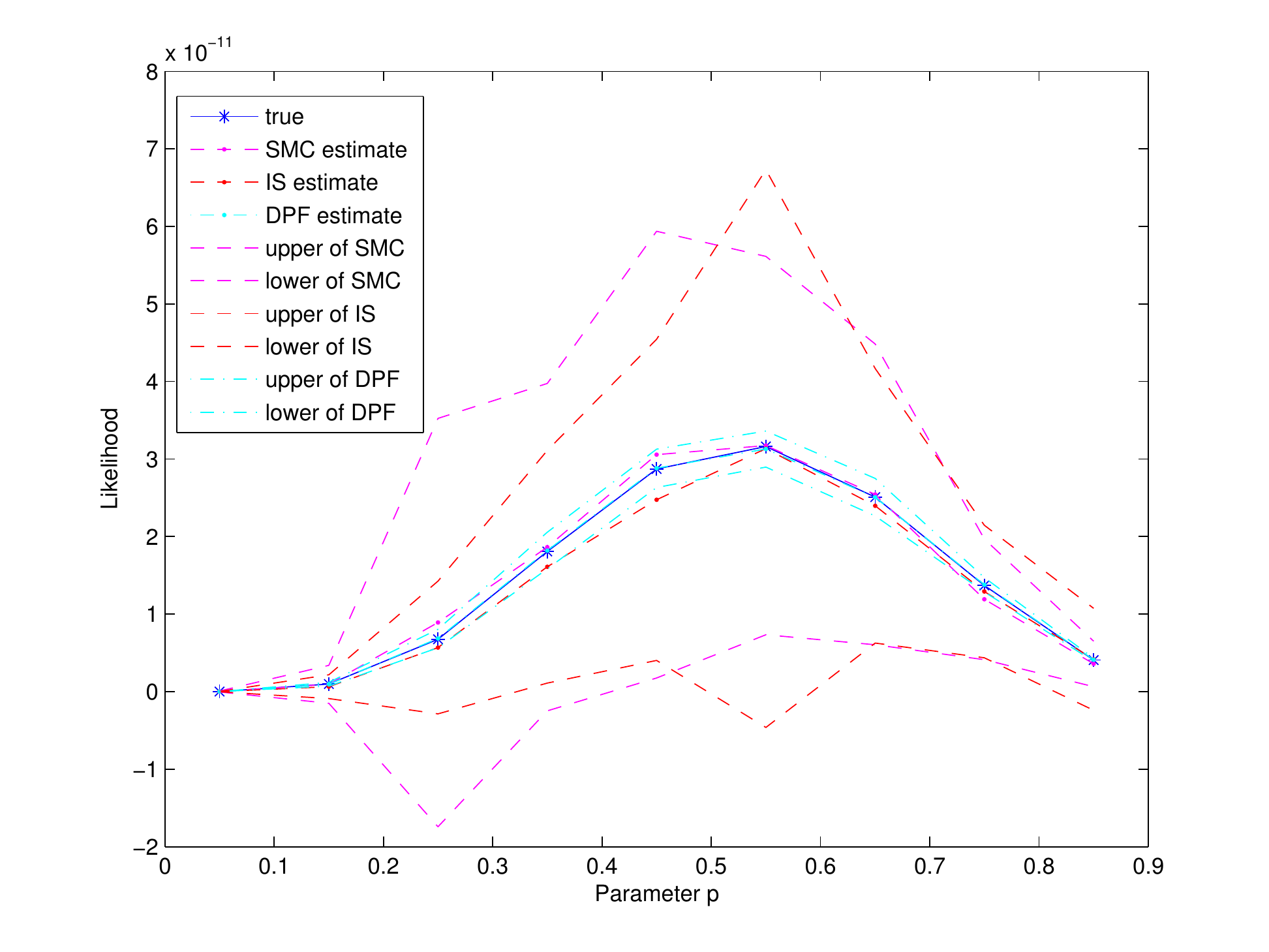}}
\caption{ \small Plot of CPU time comparison: these are results of 30 runs under IS, SMC and DPF algorithms. $N=1000$ for IS and DPF algorithms; $N=550$ for SMC algorithm. The blue solid line with stars is the true likelihoood, the purple, red and light blue dashed lines with dots are the mean of 30 SMC, IS and DPF estimates respectively, and the other two purple dashed lines, two red dashed lines and two light blue dashed lines are $\bar{x}-2s$ and $\bar{x}+2s$ of SMC, IS and DPF estimates respectively.} \label{fig:ct}
\end{figure}

\subsubsection{Simulation Results for PMCMC}\label{sim:pmcmc}

We will now test the two PMCMC algorithms (i.e.~with SMC and the DPF - recall the SMC will use dynamic resampling), we will compare to samples drawn exactly (via rejection sampling) from the posterior density with those samples generated by the PMCMC algorithms. 
The latter can be achieved when the observed network is small and we generate such a data-set. 
It is also possible to run the idealized MCMC algorithm that just samples $\theta$, which
is the best that either PMCMC algorithm can do; we also consider this in our simulations.
We generate a network model with 8 nodes and $p=0.66$.  We set only $p$ unknown with
a uniform prior on $[0,1]$. We use a proposal that is a random walk on the logit scale.
The results can be found in Figures \ref{fig:pmcmc1}-\ref{fig:pmcmc2}.

In Figures \ref{fig:pmcmc1}-\ref{fig:pmcmc2}, we can observe that the two PMCMC algorithms
perform similiarly to the marginal MCMC. In addition, they produce solutions consistent with i.i.d.~sampling. This is example is quite simple, but illustrates that such methodology can be useful for network models. A more challenging example can be found in the next Section.

\begin{figure}[!tpb]
\centering
\subfigure[Marginal MCMC]
{\includegraphics[width=6.5cm,height=4cm]{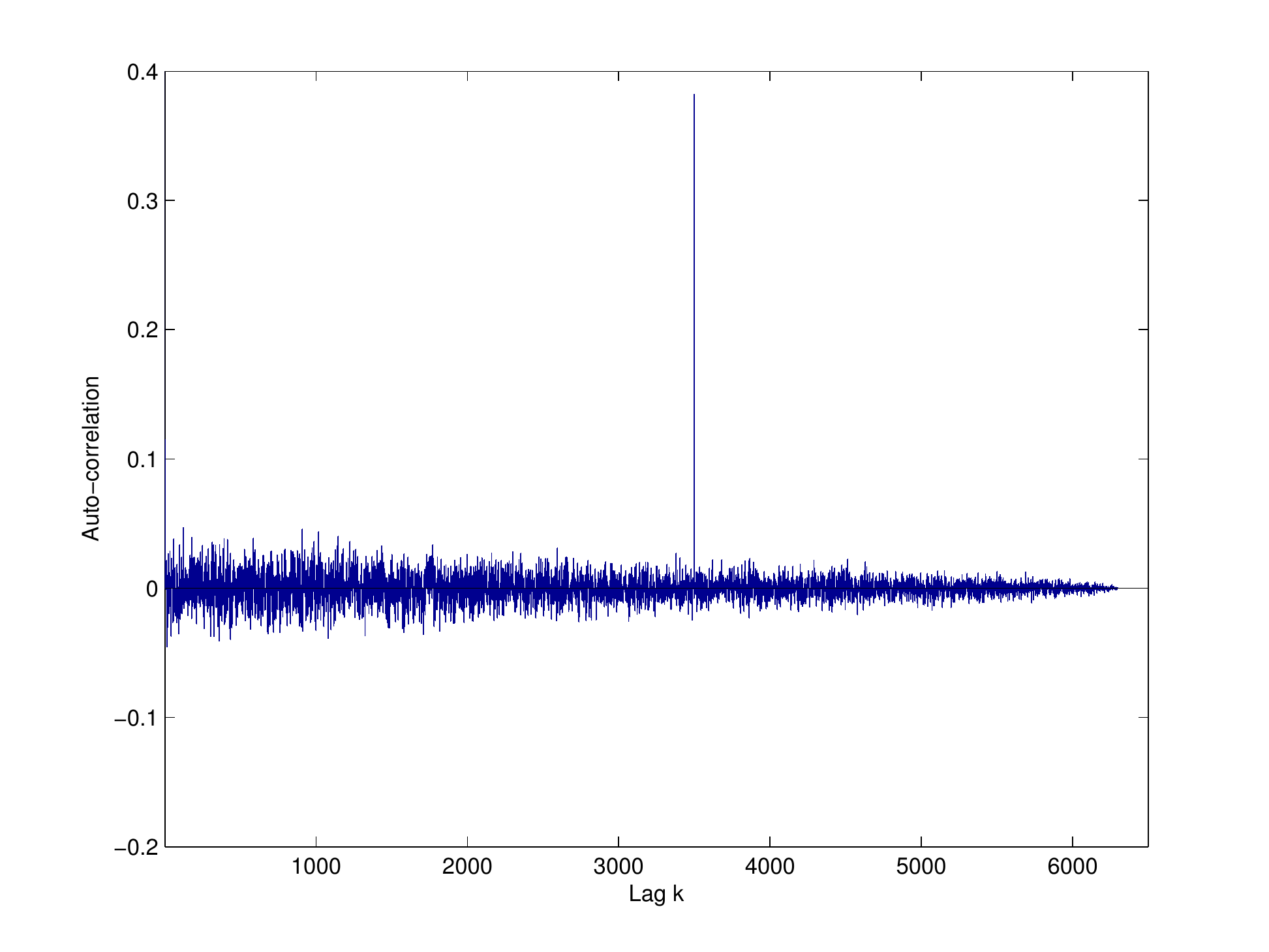}}
\subfigure[PMCMC with SMC]
{\includegraphics[width=6.5cm,height=4cm]{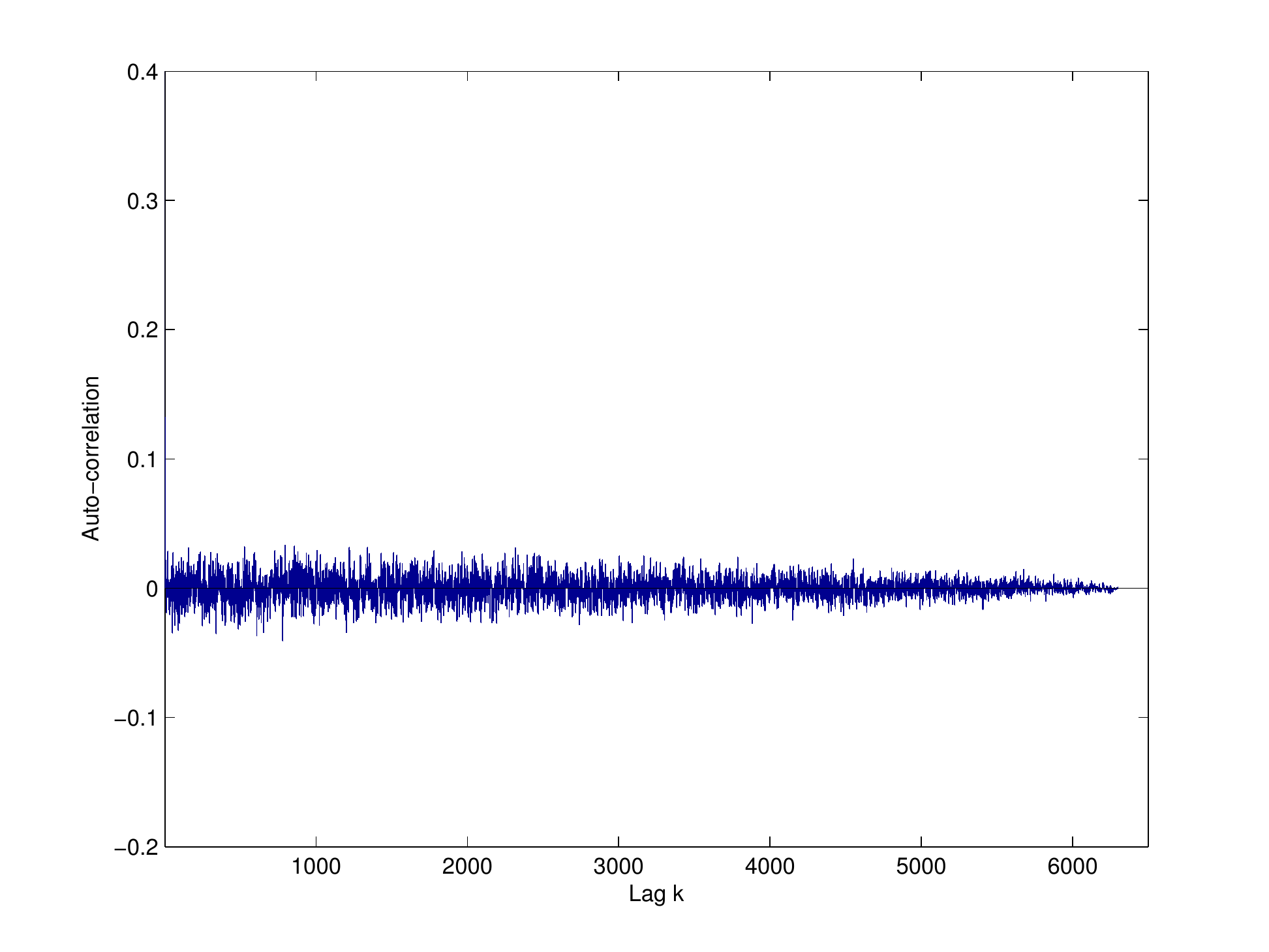}}
\subfigure[PMCMC with DPF]
{\includegraphics[width=6.5cm,height=4cm]{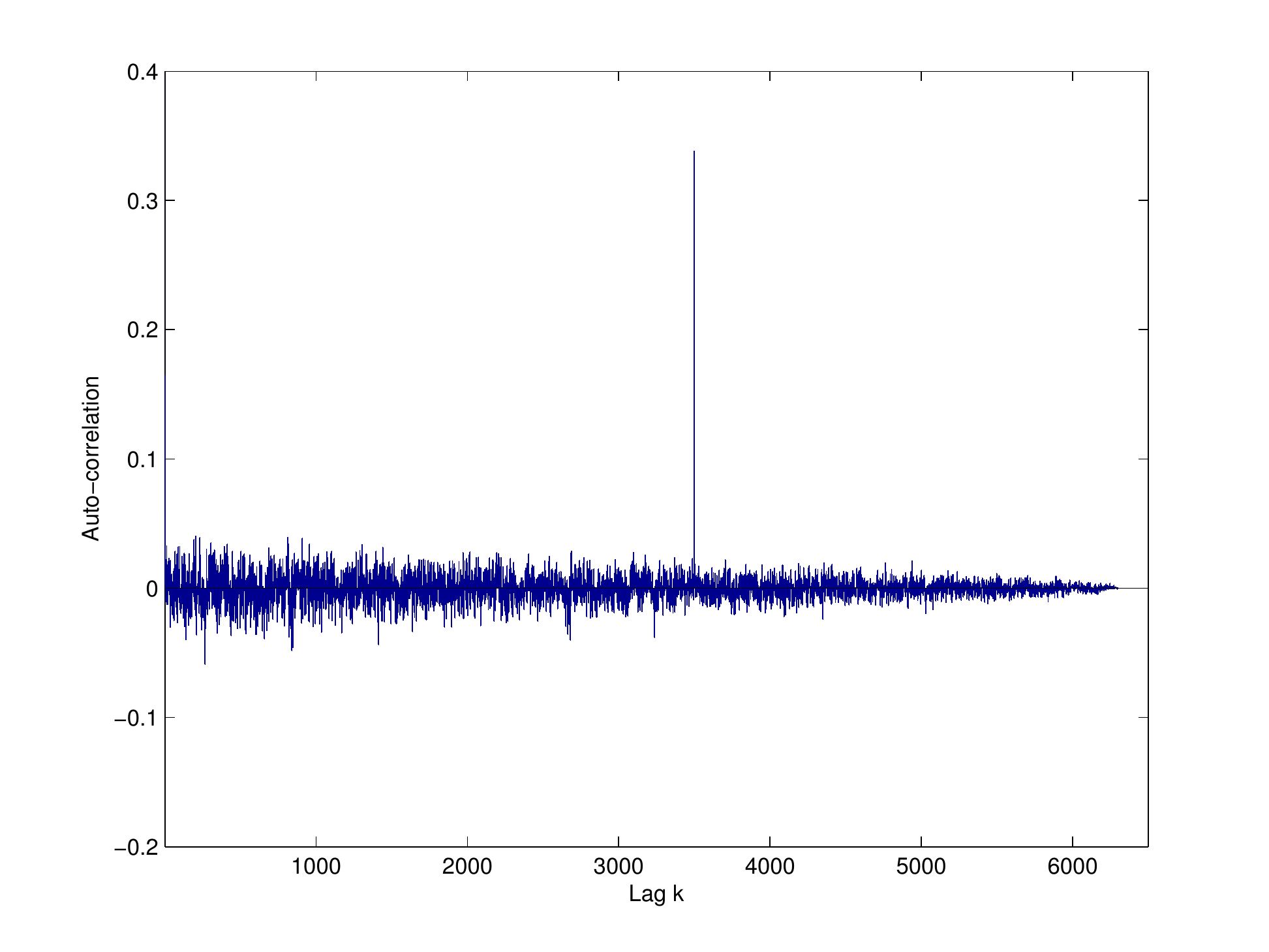}}
\caption{ \small Auto-correlation  Plots for three MCMC Algorithms (small network).}
\label{fig:pmcmc1}
\end{figure}

\begin{figure}[!tpb]
\centering
\subfigure[IID Sampling]
{\includegraphics[width=6.5cm,height=4cm]{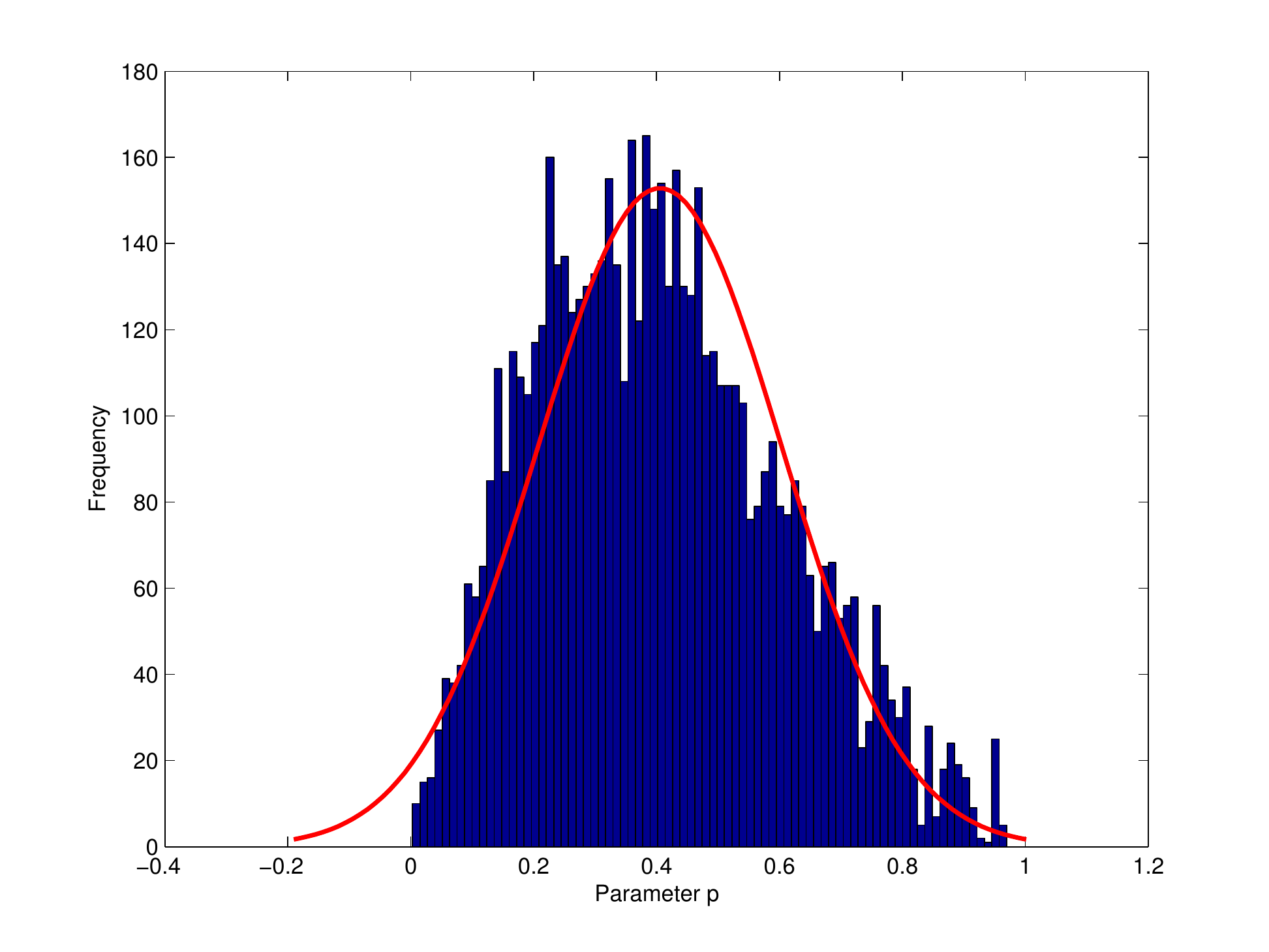}}
\subfigure[Marginal MCMC]
{\includegraphics[width=6.5cm,height=4cm]{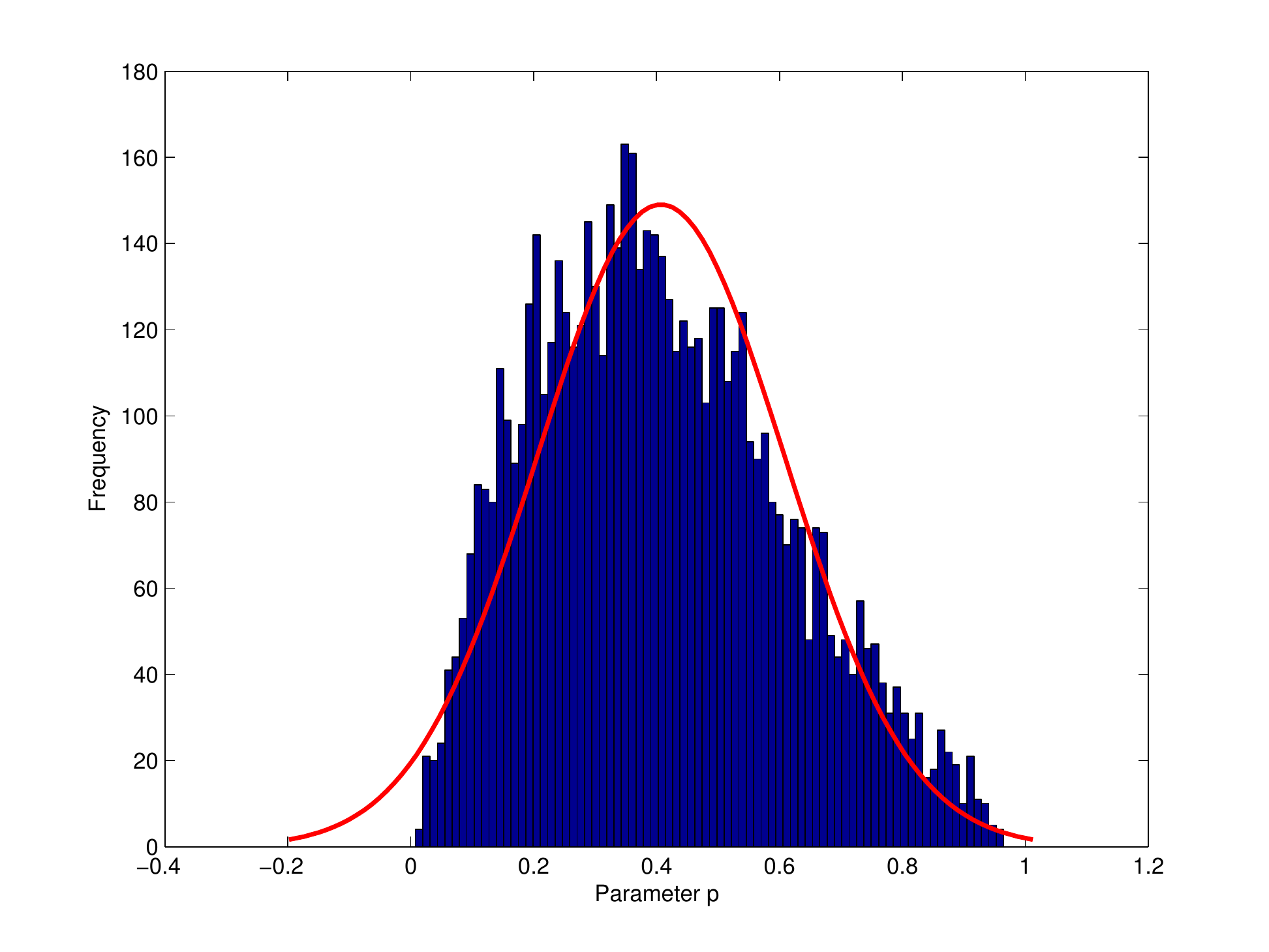}}
\subfigure[PMCMC with SMC]
{\includegraphics[width=6.5cm,height=4cm]{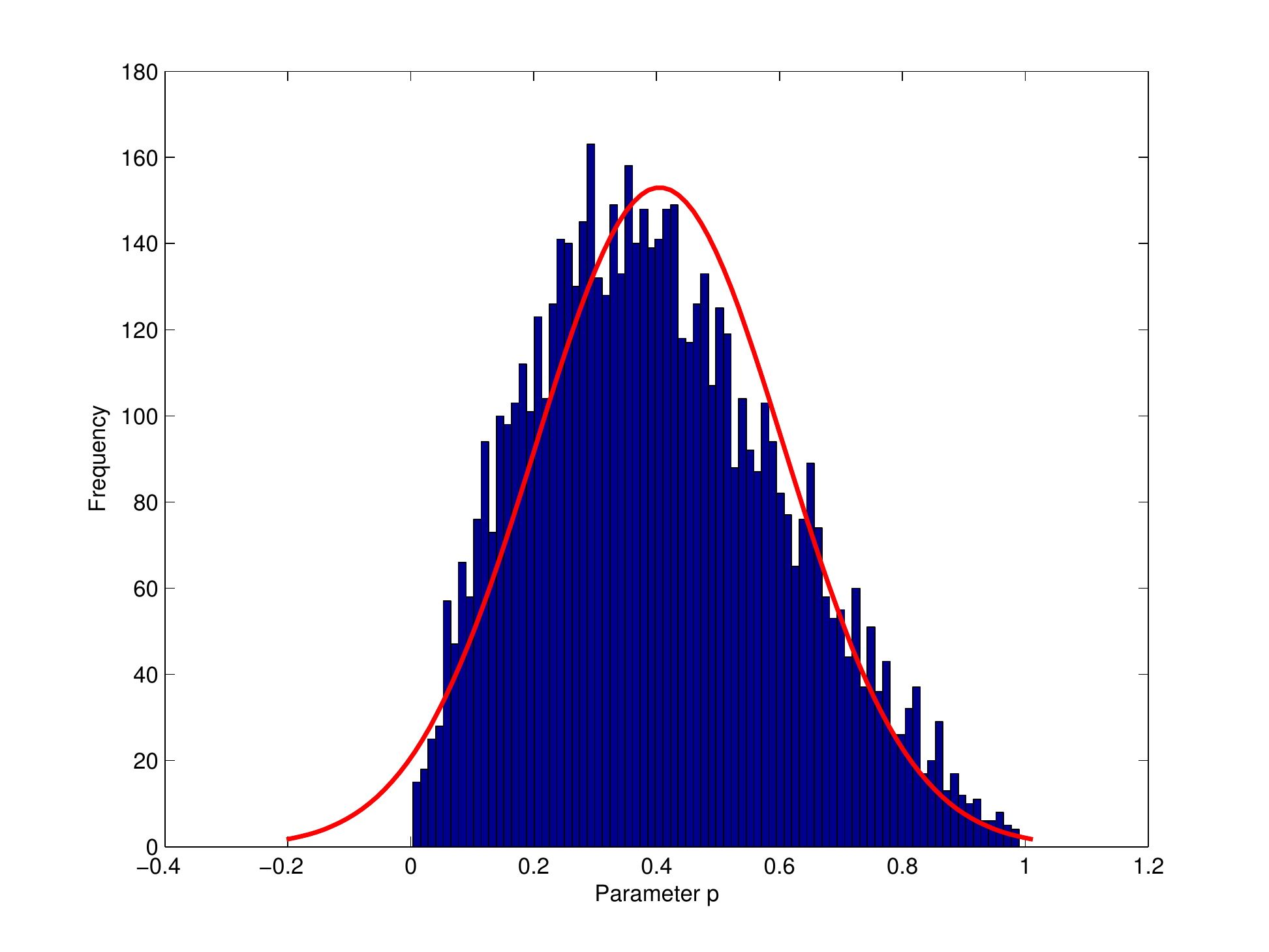}}
\subfigure[PMCMC with DPF]
{\includegraphics[width=6.5cm,height=4cm]{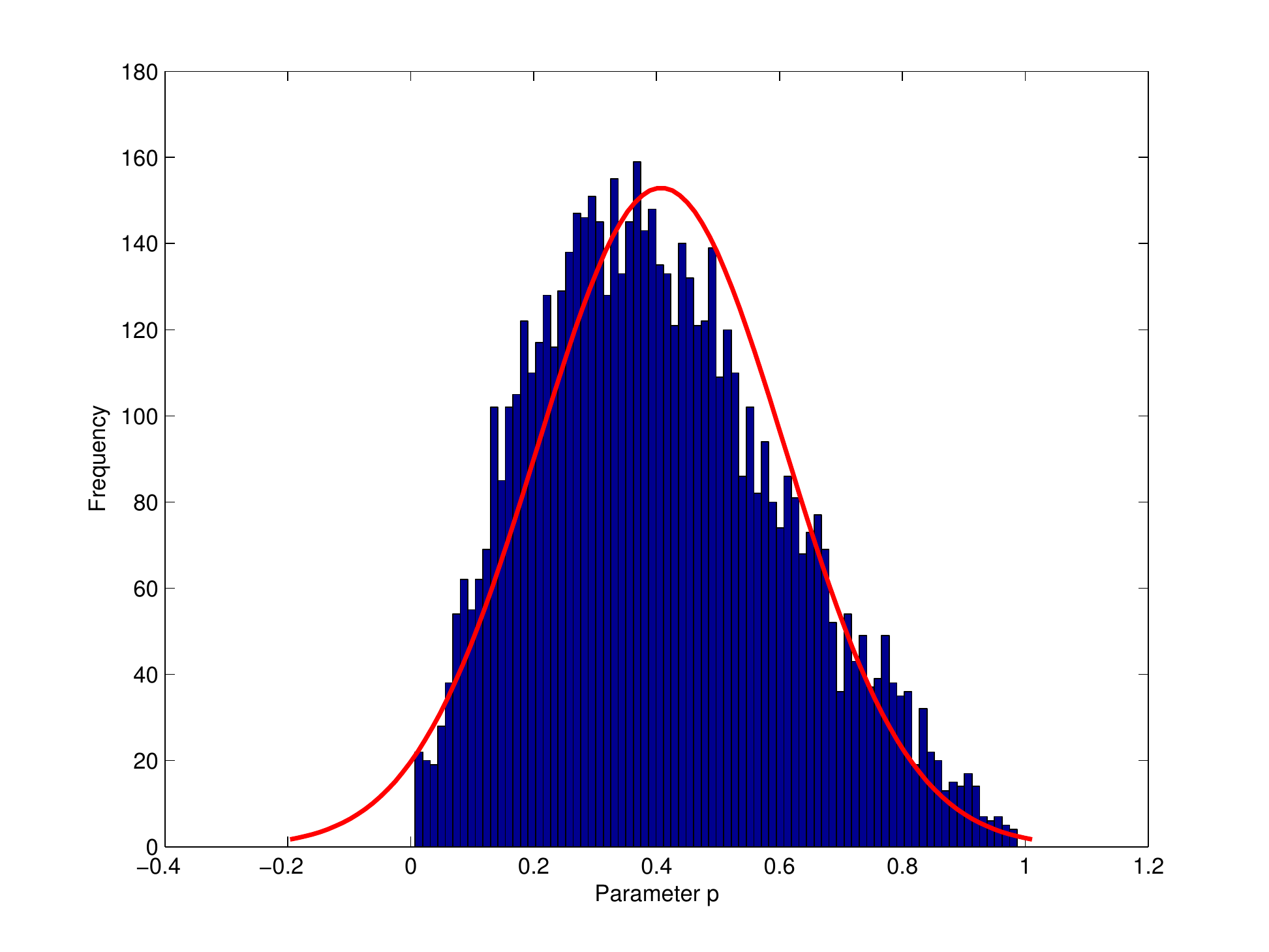}}
\caption{ \small Density Plots for IID sampling and three MCMC Algorithms.}
\label{fig:pmcmc2}
\end{figure}

\subsection{A Larger Network}

Here we consider the application of the above methods for larger sized data. We apply these methods to a graph with 100 nodes which is generated by the DA model with parameter $\theta=(1,0.66,0.33,0)$. 99 nodes are removable.

\subsubsection{Results for the Likelihood Approximation}

We now apply the three methods to approximate the likelihood of the graph. We set $\theta_0=\theta=(1,0.66,0.33,0)$ in the IS and SMC algorithms. 

Firstly, the results of the IS method are displayed in Figure \ref{fig:is1001000}; we consider $N\in\{100,1000\}$. From Figure \ref{fig:is1001000}, we see that for some parameters $p$, the corresponding log-likelihood 
is not properly estimated as the importance weights are too variable and the values of the log-weights become very small. This is supported by examining the ESS values, which for either value
of $N$ never exceeds 7 and is often 1 (the results are not displayed). The overall wall-clock computation time was 3578.41 and 40374.96 seconds for $N=100$ and $N=1000$.

Secondly, Figures \ref{fig:smc100010}-\ref{fig:smcess} display results of the SMC method (dynamic resampling), we consider again $N\in\{100,1000\}$. 
Figure \ref{fig:smc100010} shows that the variance issues of the IS method is dealt with (here $N=1000$, but similar results are obtained when $N=100$).
Figure \ref{fig:smcess} displays the values of ESS and UN. In general the algorithm performs reasonably well with regards to these criteria, with a small
issue (for this run) when there are around 25 removable nodes left. At this stage it appear that the path degeneracy effect is taking hold, but the algorithm does not collapse
and the results are reasonably reliable (recall that Proposition \ref{prop:rel_var} implies one would want $N$ to be close to $99^3$, so the performance seems quite good given this).
The overall computation time was 31932.43 and 347146.19 seconds when $N=100$ and $N=1000$. Whilst this is substantially larger than IS, the case when $N=100$ is a similar time
to IS when $N=1000$, but the former results are quite acceptable, whereas this is not the case for the latter.

Lastly, we also apply the DPF method; the results when $N=100$ is given in Figure \ref{fig:dpf10010}. The results are similar to the results of the SMC method, also quite satisfactory. 
The overall computation time was 414233.84 seconds, but in contrast to the SMC algorithm, all the particles before resampling are unique and so one would expect better results in general.
The computational time when $N=100$ is over 10 times that number when $N=100$ for the SMC.

\begin{figure}[!tpb]
\centering
\centerline{\includegraphics[width=15.5cm,height=6cm]{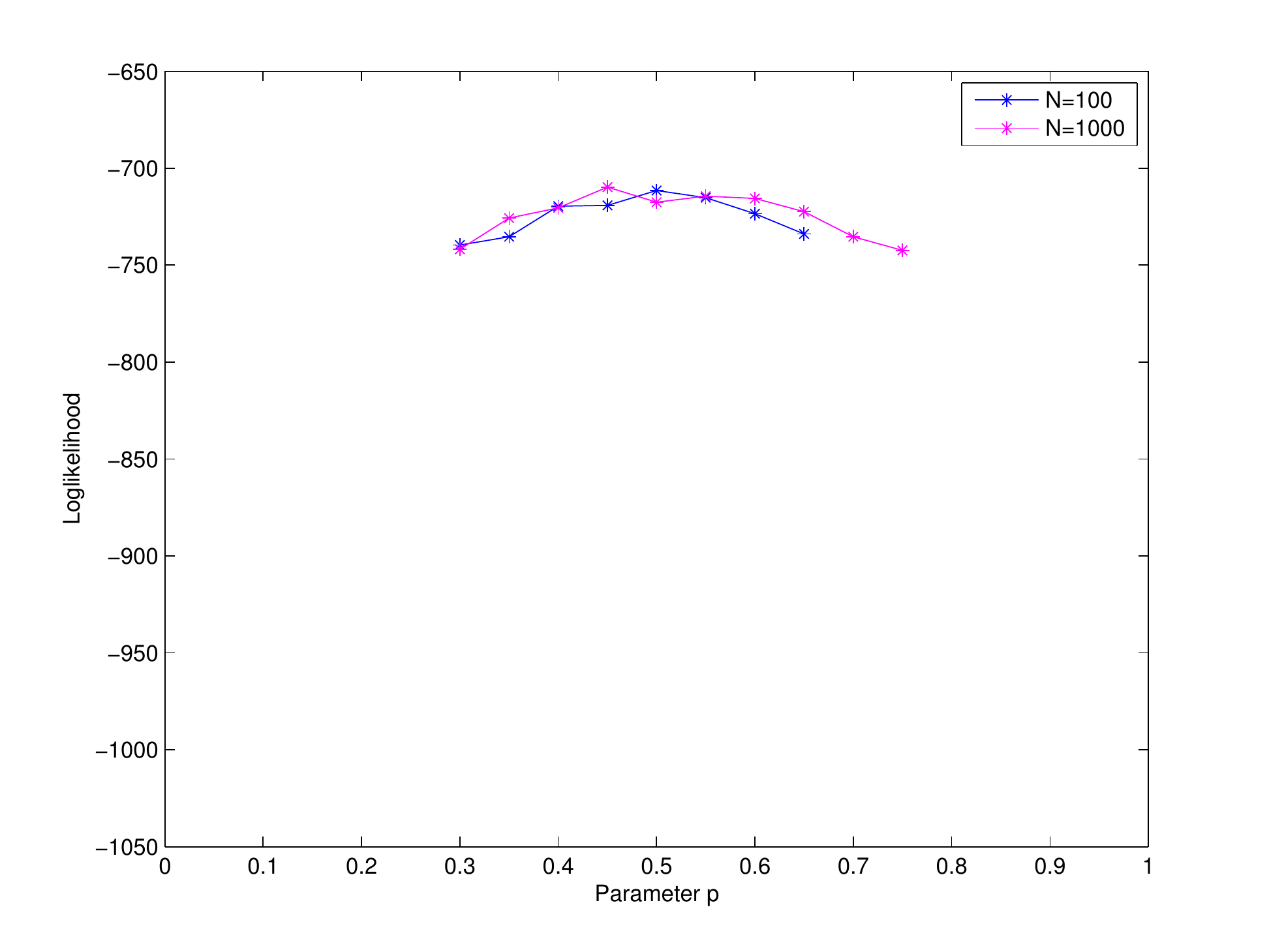}}
\caption{ \small Estimated loglikelihood curve of a single IS run under $N=100$ and $N=1000$. 
This is for the larger network.} \label{fig:is1001000}
\end{figure}

\begin{figure}[!tpb]
\centering
\centerline{\includegraphics[width=15.5cm,height=6cm]{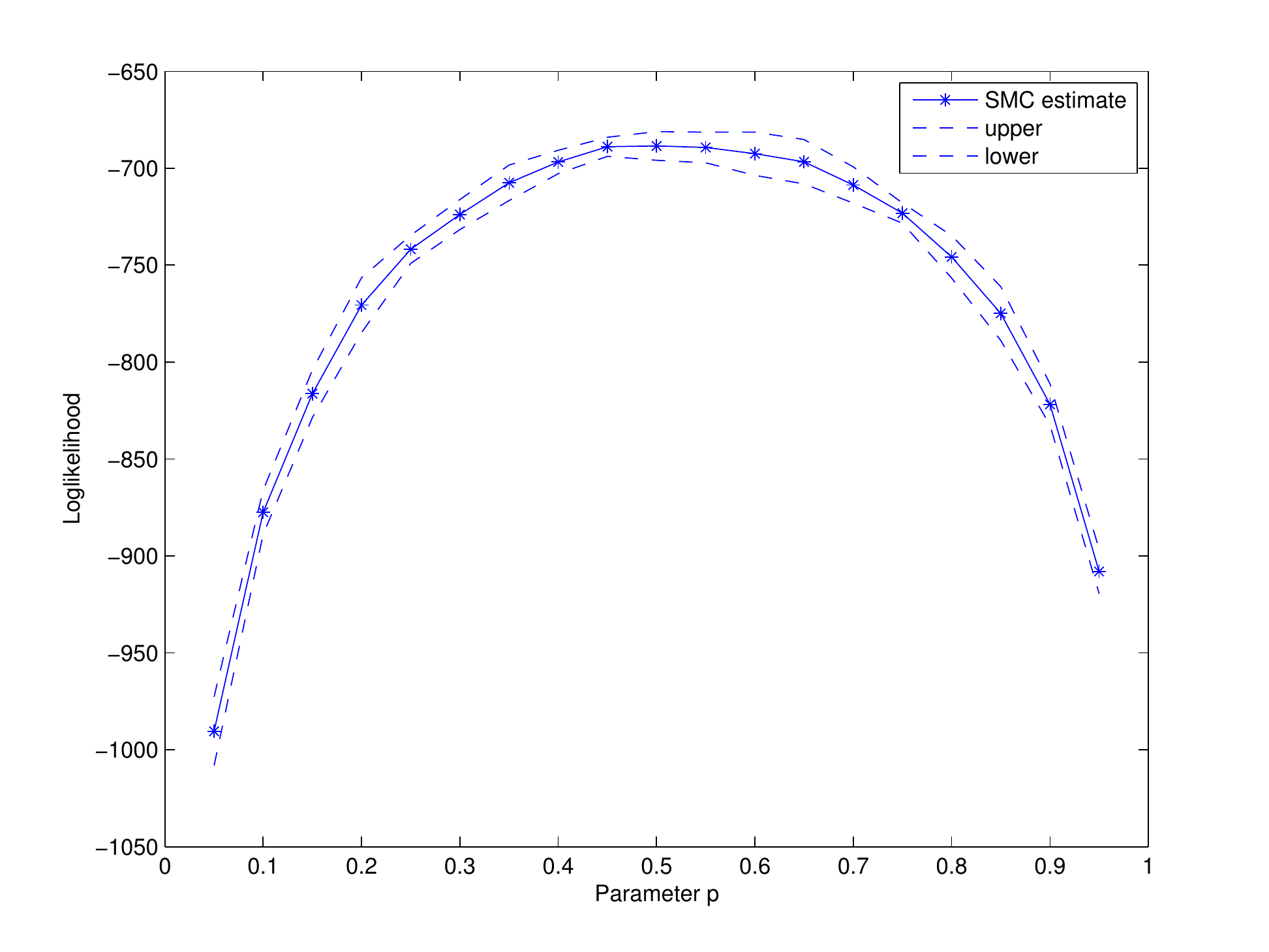}}
\caption{ \small Estimated loglikelihood curve of 50 SMC runs under $N=1000$. 
The dashed lines are $\bar{x}-2s$ and $\bar{x}+2s$ respectively (across the runs).
This is for the larger network.} \label{fig:smc100010}
\end{figure}

\begin{figure}[!tpb]
\centering
\centerline{\includegraphics[width=15.5cm,height=6cm]{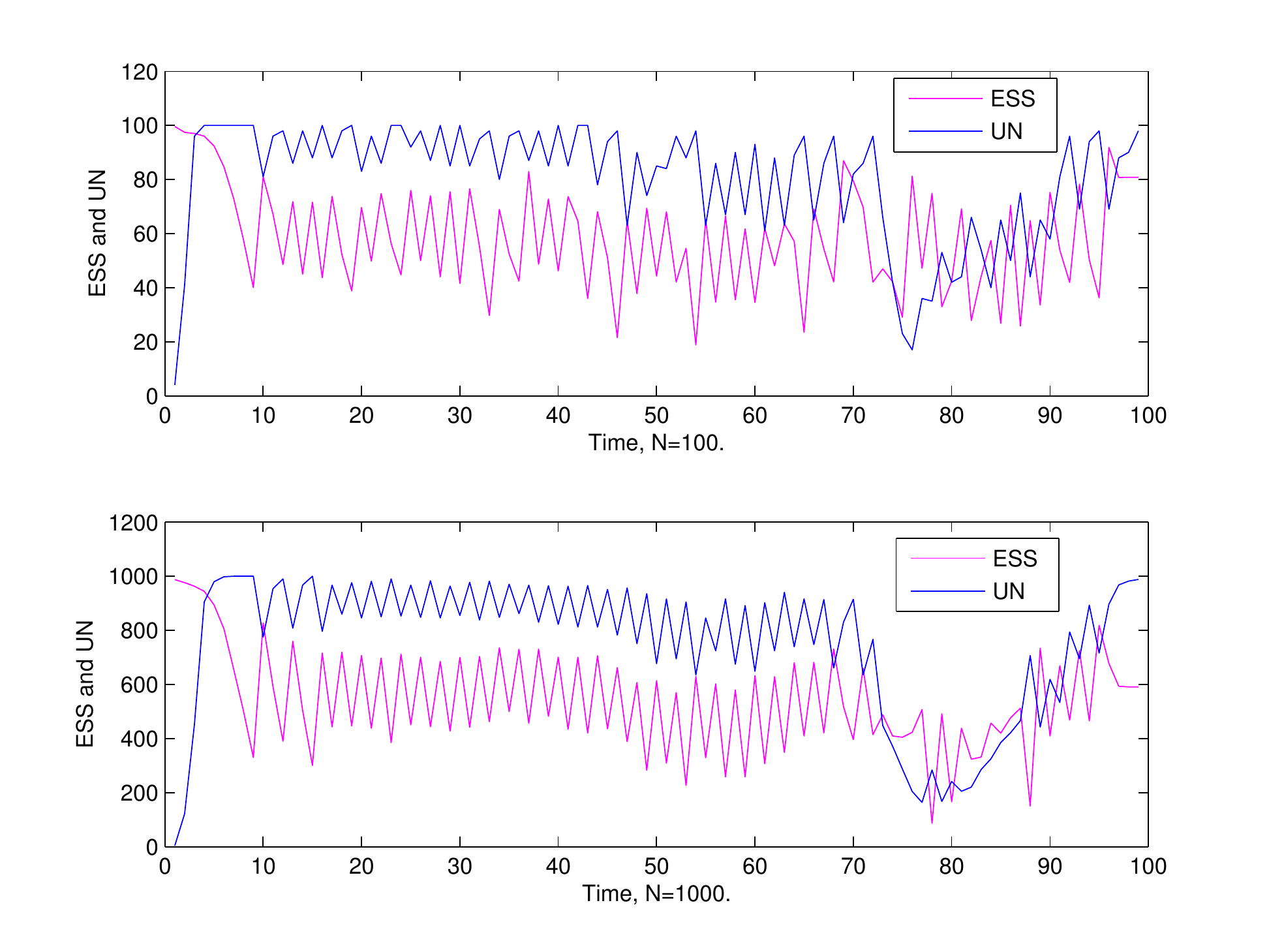}}
\caption{ \small Plots of ESS and UN for a single SMC run at every time, with $\theta=(1,0.55,0.33,0)$ and $\theta_0=(1,0.66,0.33,0)$. This is for the larger network.} \label{fig:smcess}
\end{figure}

\begin{figure}[!tpb]
\centering
\centerline{\includegraphics[width=15.5cm,height=6cm]{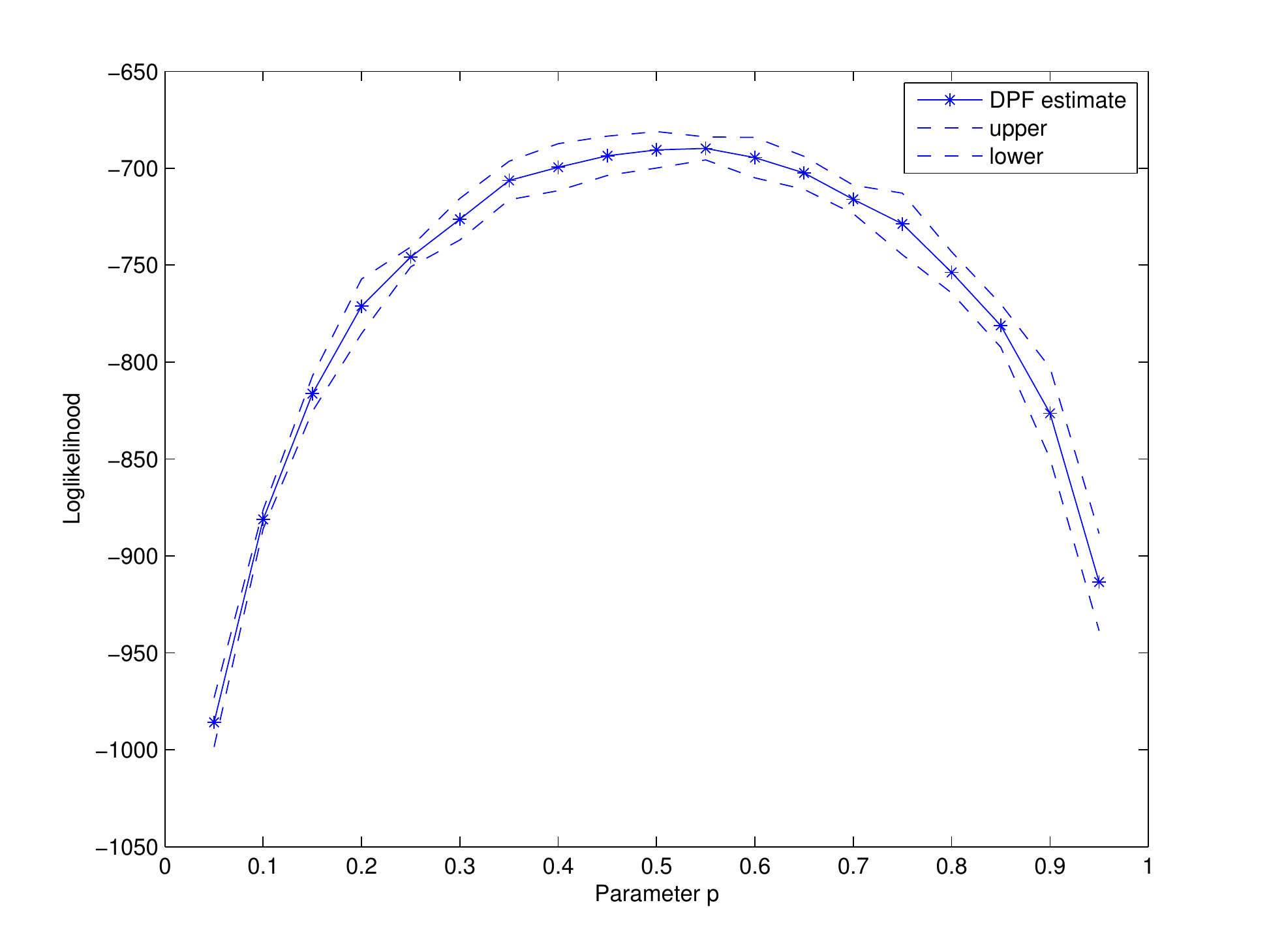}}
\caption{ \small Estimated loglikelihood curve of 50 DPF runs under $N=100$. 
The dashed lines are $\bar{x}-2s$ and $\bar{x}+2s$ respectively (across the runs).
This is for the larger network.} \label{fig:dpf10010}
\end{figure}

\subsubsection{Results for the PMCMC algorithm}

We now consider parameter inference for $p$, when it has a uniform prior on $[0,1]$.
We run two PMCMC algorithms, one which uses SMC and the other which uses a \emph{combination} of the DPF and SMC (as discussed in Section \ref{sec:dpf}). The reason for using the combination
approach is due to the computation time of applying the DPF in each iteration, which was relatively too long to just using SMC. We display results (for a typical run) for both procedures 
in Figure \ref{fig:smcanalysis}; we ran 5000 iterations with a 500 iteration burn-in and display every fifth sample. We use $N=100$ for the SMC with this value adjusted for the combination (of the DPF and SMC) to allow
roughly the same computation time, which was about 1 week.

In Figure \ref{fig:smcanalysis} we can see that the SMC + DPF combination appears to mix marginally better than the SMC when the computational time is close to the same. Given the previous
empirical results in this article, this is to be expected. The results appear to be quite reliable, in the sense that we have run multiple independent chains with similar out-puts. However,
we must take into account that the parameter space is very low-dimensional and the computation time has been quite long.

\begin{figure}[!tpb]
\centering
\subfigure[SMC]
{\includegraphics[width=6.5cm,height=6cm]{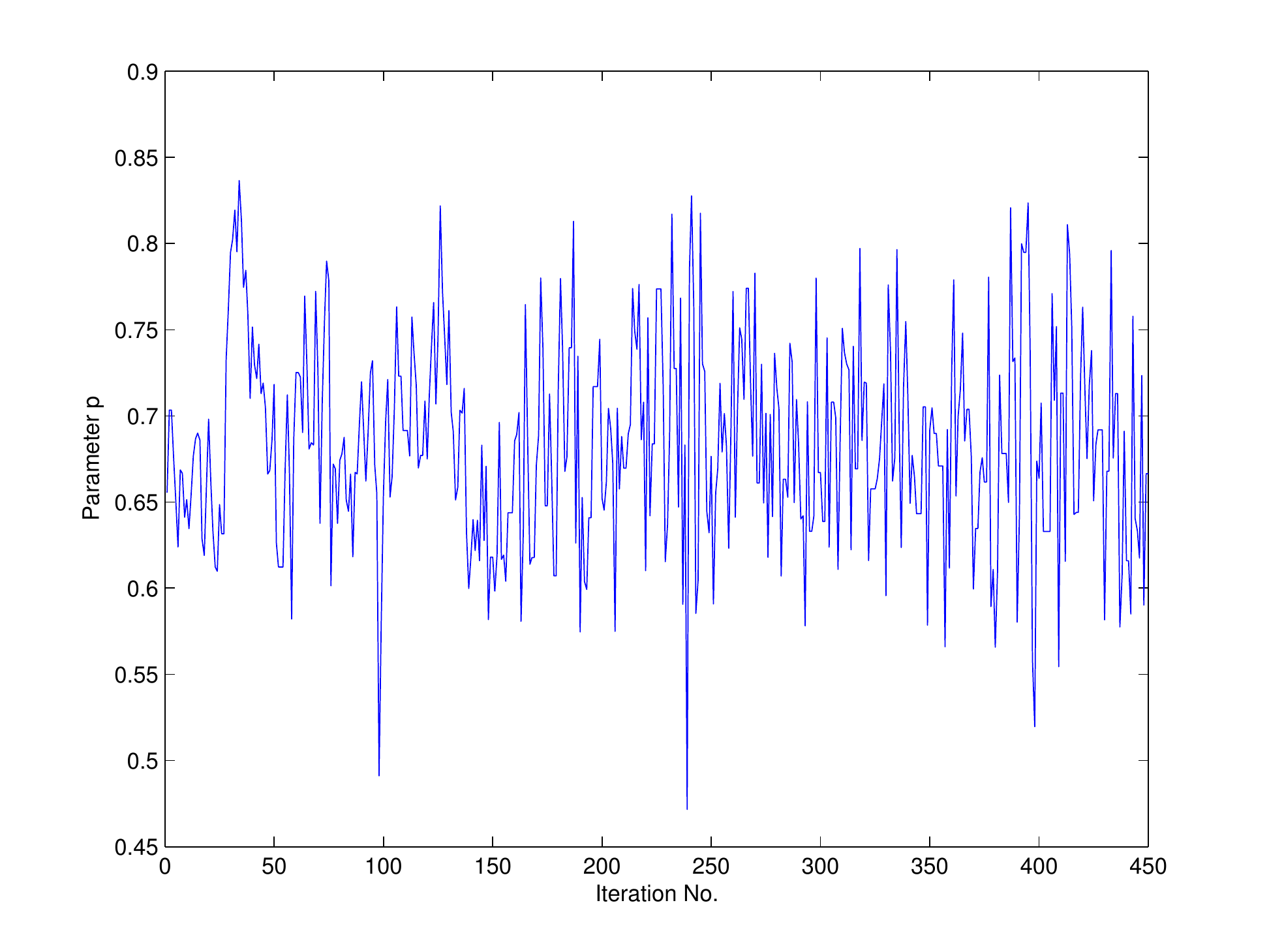}} 
\subfigure[SMC]
{\includegraphics[width=6.5cm,height=6cm]{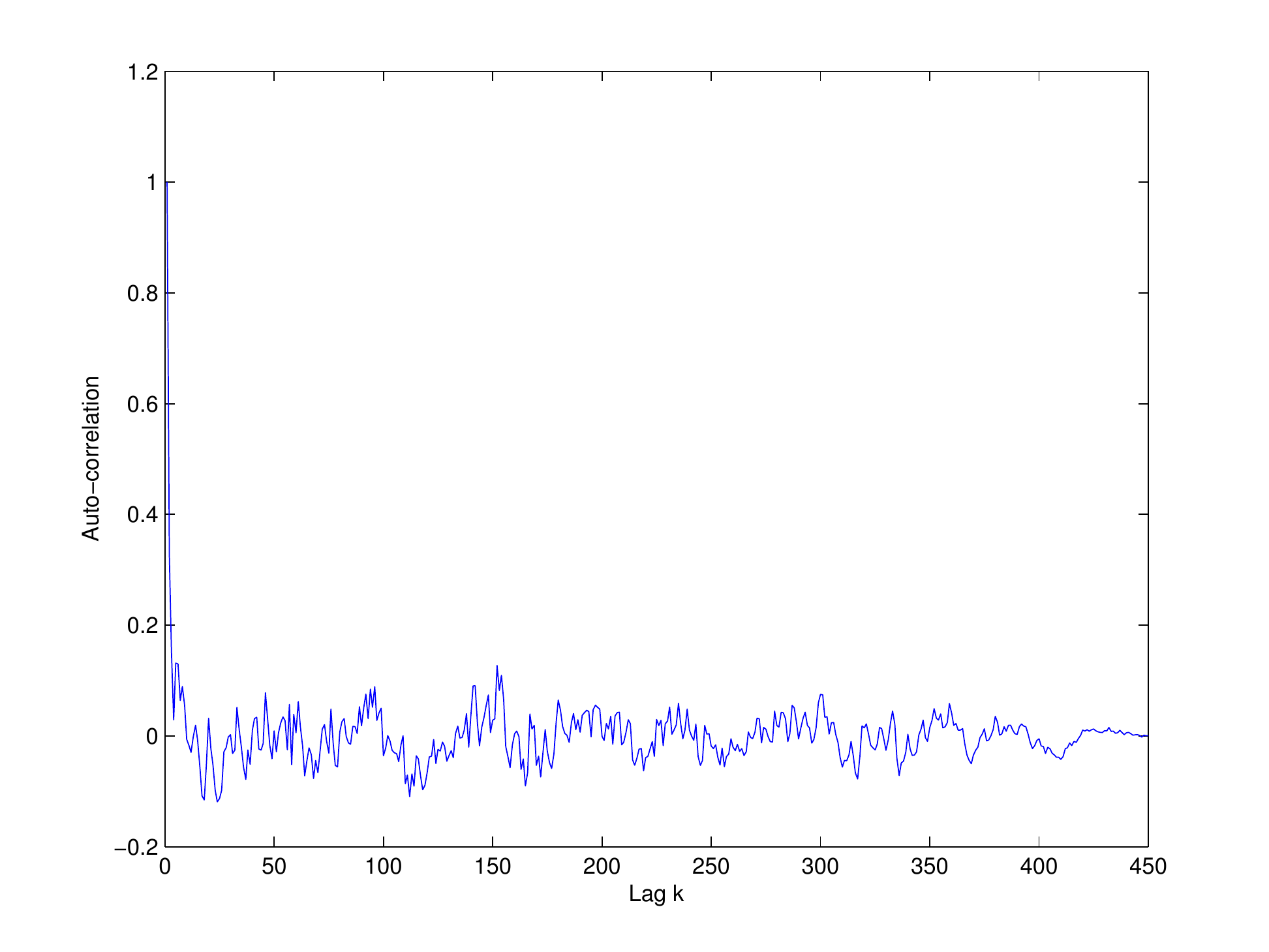}}
\subfigure[DPF + SMC]
{\includegraphics[width=6.5cm,height=6cm]{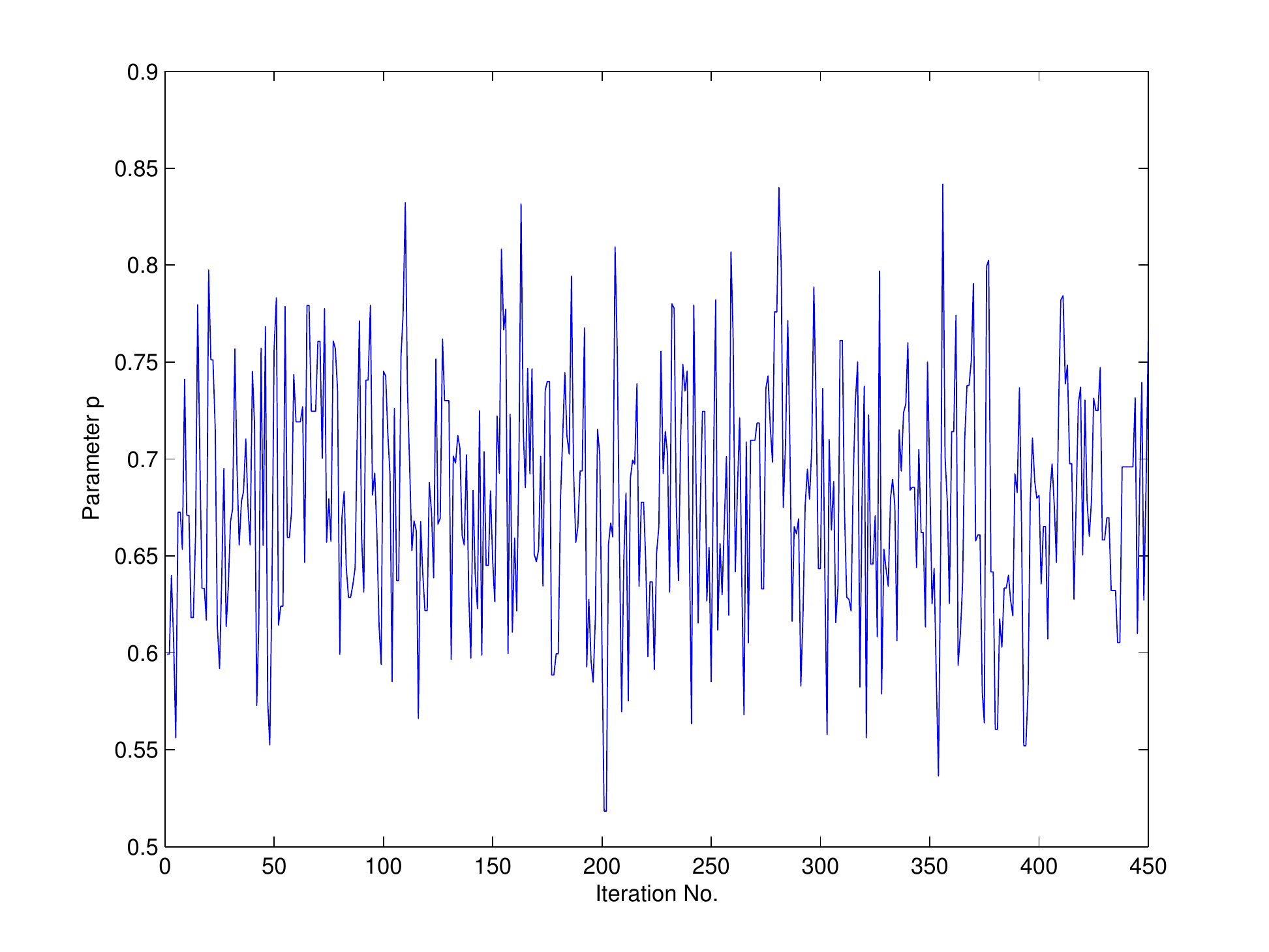}} 
\subfigure[DPF +SMC]
{\includegraphics[width=6.5cm,height=6cm]{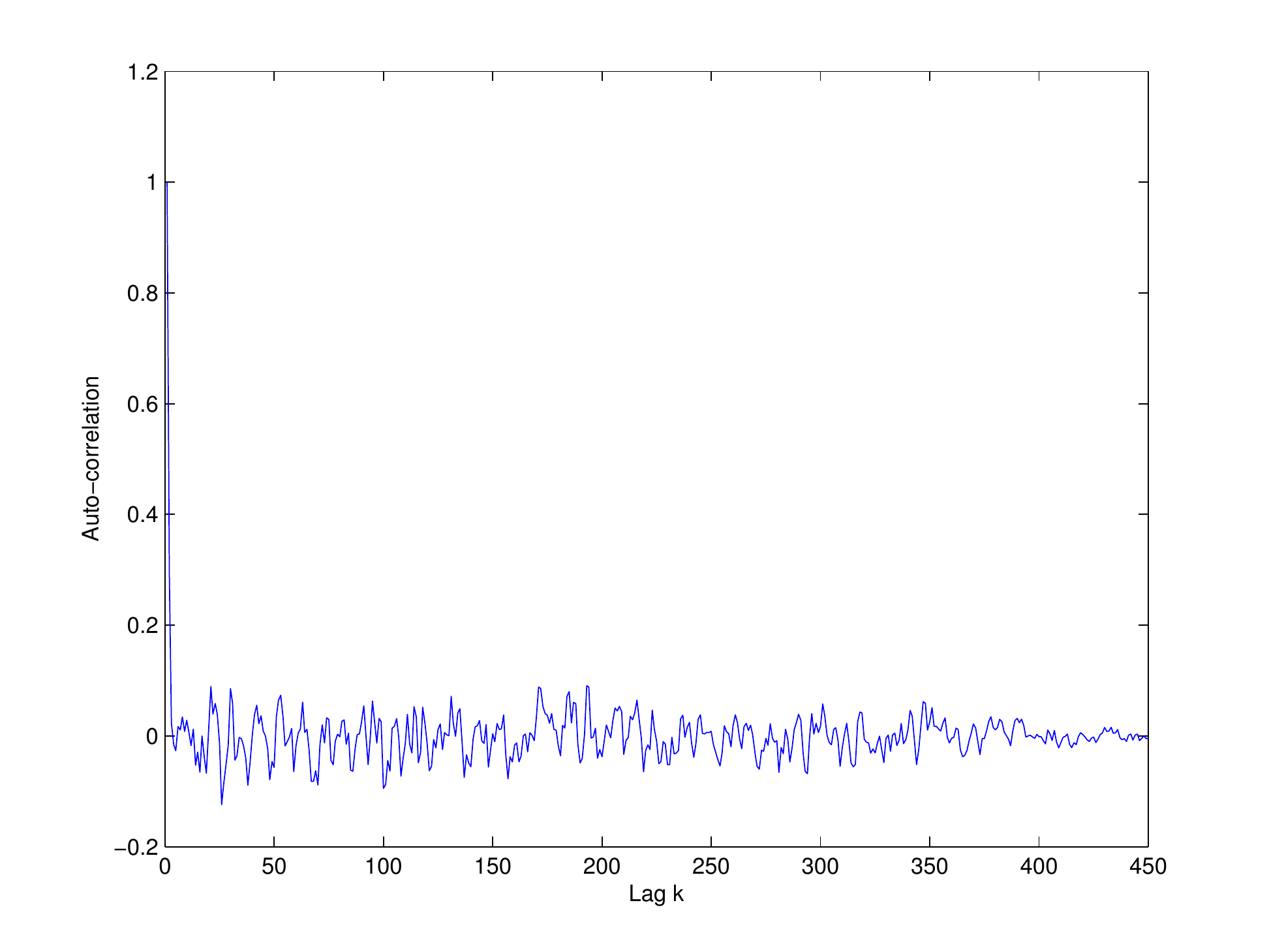}}
\caption{ \small PMCMC plots. 5000 iterations are run with a 500 iteration burn-in and we display every fifth sample.  This is for the larger network.}  
\label{fig:smcanalysis}
\end{figure}

\section{Summary}\label{sec:summary}

In this article we have considered computational methods for network models. We considered two extensions to IS 
for estimating the likelihood (for a given parameter)
for a class of network models; namely SMC and the DPF. It was then shown how these algorithms can be embedded into MCMC to perform paramater inference.
As the relative variance of the IS estimate of the likelihood typically grows at an exponential rate in the number of removable nodes, this was the main motivation
for using the two alternative approaches. We have shown that the relative variance of the SMC method will only grow at a polynomial rate in the number removable nodes.
We then illustrated these methods on small to medium sized networks and showed that the DPF and DPF inside MCMC seemed to perform better versus the SMC based versions.
In general, however, the computational time was much higher and this value was quite high for each of our algorithms.

Future work of interest is as follows. We have shown that the SMC/MCMC methods of this article seem to work well for small to medium size networks.
However, for larger networks, both memory and computational demands increase which makes it less attractive to implement these exact computational methods.
Whilst there are computational tricks to help implementation Jacob et al.~(2013) or Lee et al.~(2010) it may be preferable to make principled statistical approximations of the models (for example as in Jasra et al.~(2011))
to reduce the computational burdens. This is something that we are currently investigating.

\subsubsection*{Acknowledgements}

The second author was supported by an MOE Singapore grant.

\appendix

\section{Relative Variance Result}

In this appendix, we omit all reference to the parameter $\theta$ and we assume that our assumptions are global w.r.t~$\theta$, which will typically imply that $\Theta$ is compact.
Recall that the algorithm resamples at each time and we will assume that this is using the multinomial method.
The proof of our result relies heavily on the work in C\'erou et al.~(2011) and in order to easily verify our proof, we will use the Feynman-Kac notations in that article.
We introduce the Markov kernel $M_k:\mathsf{V}_{t-k+1:t}\rightarrow\mathcal{P}(\mathsf{V}_{t-k:t})$, $1\leq k \leq t-t_0-1$, where $\mathcal{P}(\mathsf{V})$ are the collection of probabilities on a set $\mathsf{V}$.
The Markov kernel is, for $x_{k-1}\in \mathsf{V}_{t-k+1:t}$, $x_k\in \mathsf{V}_{t-k:t}$, $x_k=(\tilde{x}_{k},x_k')\in \mathsf{V}_{t-k+1:t}\times\mathsf{W}_{t-k}(\tilde{x}_{k})$
$$
M_k(x_{k-1},x_k) = \delta_{x_{k-1}}(\tilde{x}_{k}) q(x_k'|\tilde{x}_{k})
$$
where the conditional probability $q$ is as described in the algorithm of Section \ref{sec:smc} (always assumed to be positive, when $x_k'\in\mathsf{W}_k(\tilde{x_k})$). Further we use $G_k:\mathsf{V}_{t-k:t}\rightarrow \mathbb{R}^+$ $0\leq k \leq t-t_{0}-1$, to denote the importance weights $w_k(\cdot)$ in Section \ref{sec:smc}.
We make the following hypotheses. Note that if $x\in\mathsf{V}_{t-k+1:t}$ and one makes a transition via $M_k$, then the produced state is $y=(x,y')$, with $y'\in\mathsf{W}_k(x)$.

\begin{hypA}\label{assumption}
For $t-t_0$ fixed, there exist a $0<\xi(t-t_0)<+\infty$ such that
$$
\sup_{0\leq k \leq t-t_0-1} \sup_{(x,y)\in\mathsf{V}_{t-k:t}^2} \frac{G_k(x)}{G_k(y)} \leq \xi(t-t_0).
$$
For each $1\leq k \leq t-t_0-1$, and $t-t_0$ fixed there exists a $\xi_k(t-t_0)\in[1,\infty)$ such that for any $(x,u) \in \mathsf{V}_{t-k+1:t}^2$, $(y',v')\in \mathsf{W}_k(x)\times\mathsf{W}_k(u)$, $(y,v) = ((x,y'),(u,v'))$
$$
M_k(x,y) \leq \xi_k(t-t_0) M_k(u,v).
$$
\end{hypA}
The assumption simply says that the importance weights are lower-bounded away from zero and upper-bounded, and that these bounds are uniform in the time parameter; the bound can depend upon the number of removable nodes, which happens in practice. In addition, the assumption on $M_k$ is quite reasonable as the algorithm evolves upon a finite state-space.
We note that the constant $\xi_k(t-t_0)$ can depend upon the number of removable nodes, which, again, one might expect in practice.
We will use the notation $Q_k(x,y) = G_{k-1}(x) M_k(x,y)$, with $1\leq k \leq t-t_0-1$ and use the semi-group notation
$$
Q_{k,n}(x_k,x_n) = \sum_{x_{k+1},\dots,x_{n-1}} Q_{k+1}(x_k,x_{k+1}) \dots Q_n(x_{n-1},x_n), \quad k <n.
$$
$Q_{k,k}$ is the identity. Let $\varphi:\mathsf{V}_{t-n:t}\rightarrow\mathbb{R}$ be any bounded function, we use the notation $Q_{k,n}(\varphi)(x_k) = \sum_{x_n} Q_{k,n}(x_k,x_n)\varphi(x_n)$.

\begin{proof}[Proof of Proposition \ref{prop:rel_var}]
We simply need to show that for $k<n$, 
$$
\sup_{(x,y)\in\mathsf{V}_{t-k:t}^2} \frac{Q_{k,n}(1)(x)}{Q_{k,n}(1)(y)} \leq \xi(t-t_0)\xi_{k+1}(t-t_0)
$$
the proof will then follow directly from Theorem 5.1 and Corollary 5.2 of C\'erou et al.~(2011). Noting that
$Q_{k,n}(\varphi)(x_k) = Q_{k+1}(Q_{k+2,n}(\varphi))(x_k)$, it will suffice to show that, for any positive and bounded function $\varphi:\mathsf{V}_{t-k-1:t}\rightarrow\mathbb{R}^+$, we have
\begin{equation}
\sup_{(x,y)\in\mathsf{V}_{t-k:t}^2} \frac{Q_{k}(\varphi)(x)}{Q_{k}(\varphi)(y)} \leq \xi(t-t_0)\xi_{k+1}(t-t_0) \label{eq:prf_eq}.
\end{equation}

We note that by (A\ref{assumption}), it follows that for any $(x,y)\in\mathsf{V}_{t-k:t}^2$
\begin{equation}
\frac{Q_{k}(\varphi)(x)}{Q_{k}(\varphi)(y)} \leq \xi(t-t_0) \frac{M_{k+1}(\varphi)(x)}{M_{k+1}(\varphi)(y)} \label{eq:prf_eq1}.
\end{equation}
So, now consider for any $x\in\mathsf{V}_{t-k}$
\begin{eqnarray*}
M_{k+1}(\varphi)(x) & = & \sum_{u\in\mathsf{V}_{t-k-1:t}} M_{k+1}(x,u) \varphi(u)  \\
& = & \sum_{u\in\mathsf{V}_{t-k-1:t}} M_{k+1}(x,u)\mathbb{I}_{\{x\}\times\mathsf{W}_{k+1}(x)}(u) \varphi(u) \\
&\leq & \xi_{k+1}(t-t_0)M_{k+1}(y,v) \mathbb{I}_{\{y\}\times\mathsf{W}_{k+1}(y)}(v) \sum_{u\in\mathsf{V}_{t-k-1:t}}\mathbb{I}_{\{x\}\times\mathsf{W}_{k+1}(x)}(u) \varphi(u).
\end{eqnarray*}
Then on multiplying both sides of the inequality by $\varphi(v)$ and summing w.r.t.~$v$, we have
\begin{eqnarray*}
M_{k+1}(\varphi)(x) \sum_{v\in \mathsf{V}_{t-k-1:t}} \varphi(v)  & \leq & \xi_{k+1}(t-t_0) \bigg(\sum_{v \in \mathsf{V}_{t-k-1:t}} M_{k+1}(y,v) \mathbb{I}_{\{y\}\times\mathsf{W}_{k+1}(y)}(v) \varphi(v) \bigg)\times \\ & & \sum_{u\in\mathsf{V}_{t-k-1:t}}\mathbb{I}_{\{x\}\times\mathsf{W}_{k+1}(x)}(u) \varphi(u) \\
& \leq & \xi_{k+1}(t-t_0) M_{k+1}(\varphi)(y)  \sum_{u\in\mathsf{V}_{t-k-1:t}} \varphi(u).
\end{eqnarray*}
Hence we shown that for any $(x,y)\in\mathsf{V}_{t-k:t}^2$
$$
\frac{M_{k+1}(\varphi)(x)}{M_{k+1}(\varphi)(y)} \leq \xi_{k+1}(t-t_0).
$$
On noting \eqref{eq:prf_eq1}, \eqref{eq:prf_eq} and the above arguments, the proof is completed.
\end{proof}

\end{document}